\documentclass[11pt]{article}
\usepackage[margin=1in]{geometry}
\usepackage{cite,microtype,graphicx}
\usepackage{hyperref,aliascnt}
\usepackage{amsmath,amssymb,amsfonts,amsthm}

\newtheorem{theorem}{Theorem}[section]

\newaliascnt{lemma}{theorem} \newtheorem{lemma}[lemma]{Lemma}
\aliascntresetthe{lemma} 

\newaliascnt{corollary}{theorem}
\newtheorem{corollary}[corollary]{Corollary}
\aliascntresetthe{corollary}

\newaliascnt{definition}{theorem}

\aliascntresetthe{definition}

\newaliascnt{observation}{theorem}
\newtheorem{observation}[observation]{Observation}
\aliascntresetthe{observation}

\title{The Graphs of Stably Matchable Pairs}

\author{David Eppstein\\
Computer Science Department, University of California, Irvine}

\date{ }

\begin{document}

\maketitle
\begin{abstract}
We study the graphs formed from instances of the stable matching problem by connecting pairs of elements with an edge when there exists a stable matching in which they are matched. Our results include the NP-completeness of recognizing these graphs, an exact recognition algorithm that is singly exponential in the number of edges of the given graph, and an algorithm whose time is linear in the number of vertices of the graph but exponential in a polynomial of its carving width. We also provide characterizations of graphs of stably matchable pairs that belong to certain classes of graphs, and of the lattices of stable matchings that can have graphs in these classes.
\end{abstract}

\section{Introduction}

A stable matching instance, in its most basic form, consists of two sets of elements to be matched (for instance, students and residencies), with each element in one set having preferences that can be described as a linear ordering of the elements of the other set. In a slightly more general form, the form we consider here, an element may 
prefer remaining unmatched over being matched to some undesired elements. This can be expressed by making each preference ordering be a linear ordering of a subset of elements, removing the undesired elements from this ordering.
A matching is \emph{stable} when no matched element would prefer being unmatched over its assigned match, and no pair of elements would both prefer being matched to each other over their assigned outcomes. The Gale--Shapley theorem, which applies equally well in this more general setting, states that a stable matching exists; one such matching can be found in time linear in the total length of the preference orderings by the Gale--Shapley algorithm~\cite{GalSha-AMM-62}.

A given instance of stable matching may have many stable matchings, forming a distributive lattice in which the matching found by the Gale--Shapley algorithm is the top element. Although this lattice may consist of exponentially many matchings~\cite{KarGhaWeb-STOC-18}, it has a concise description as a partially ordered set of alternating cycles, called \emph{rotations}, which can be constructed in polynomial time~\cite{IrvLea-SICOMP-86}. A given pair of elements participates in at least one stable matching if and only if it is either part of the top stable matching found by the Gale--Shapley algorithm, or part of one of these rotations. Therefore, from a given instance of stable matching, we may construct in polynomial time a bipartite graph, the graph of pairs that can be matched to each other in at least one stable matching.

In this work, we ask: can we reverse this process? Given a bipartite graph, can we determine whether it comes from an instance of stable matching in this way? Can we construct a stable matching instance that has a given graph as its graph of stably matchable pairs? What structural properties must such a graph have? For instance, by the rural hospitals theorem~\cite{McVWil-BIT-70,Rot-Econ-86} we know that every stable matching has the same set of matched elements; in graph-theoretic terms, this means that removing the isolated vertices from the graph of stably matchable pairs (the elements that cannot be stably matched) leaves a balanced bipartite graph (one with equal numbers of vertices on each side of the bipartition), and this graph is a \emph{matching-covered graph} meaning that each of its edges participates in a perfect matching. Matching-covered balanced bipartite graphs can be recognized efficiently, in the same asymptotic time as finding a single perfect matching~\cite{Reg-AAAI-94,Tas-TCS-12}. Are these necessary conditions sufficient? Is every matching-covered balanced bipartite graph the graph of stably matchable pairs of some stable matching instance?

\begin{figure}[t]
\centering\includegraphics[width=0.8\textwidth]{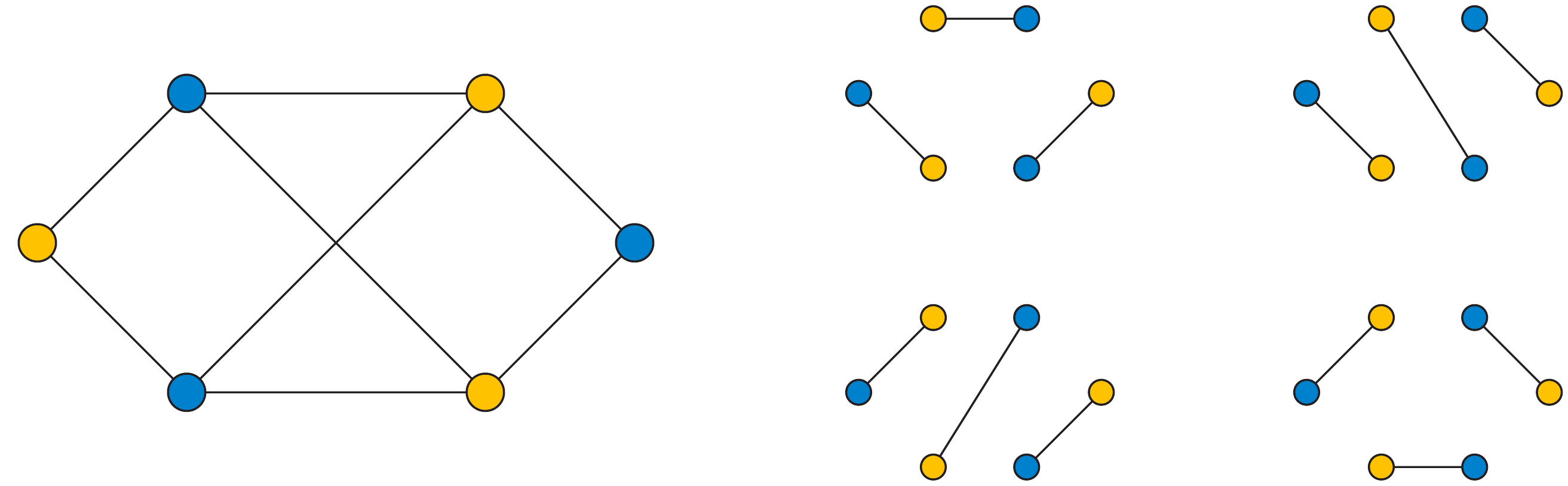}
\caption{$K_{3,3}-e$ (left) is matching-covered by its four perfect matchings (right) but it cannot be a graph of stably matchable pairs (\autoref{obs:k33-e}).}
\label{fig:k33-e}
\end{figure}

\subsection{New results}

The example of $K_{3,3}-e$ shown in \autoref{fig:k33-e} answers some of these questions: it is a balanced, matching-covered bipartite graph, but (as we will demonstrate later in \autoref{obs:k33-e}) cannot be generated as the graph of stably matchable pairs of any stable matching instance. More generally, in this paper we provide the following results:

\begin{itemize}
\item In \autoref{sec:representable} we search for simple criteria that can determine in some cases whether a given graph is a graph of stably matchable pairs. We observe in \autoref{sec:regular} that these graphs have no forbidden induced subgraphs: every bipartite graph is an induced subgraph of a larger regular graph, and every regular graph is a graph of stably matchable pairs. We prove in \autoref{thm:eq-deg-match} that, for a subcubic graph to be a graph of stably matchable pairs, it is necessary that its degree-one vertices and that its degree-three vertices induce subgraphs that have perfect matchings, and sufficient that in addition its degree-two vertices induce a subgraph that has a perfect matching. We prove in \autoref{thm:outerplanar} that an outerplanar graph is a graph of stably matchable pairs if and only if it has no articulation vertex, and we characterize in \autoref{thm:grid} the rectangular grid graphs that are graphs of stably matchable pairs. More generally, \autoref{thm:product} concerns realizability of Cartesian product graphs. 

\item In \autoref{sec:lattice} we study the distributive lattices of stable matchings associated with graphs of stably matchable pairs, and we examine how the structure of the graph affects the structure of the lattice. We prove in \autoref{thm:dl2sub3} that every finite distributive lattice is the lattice of stable matchings of an instance whose graph of stably matchable pairs is subcubic, and in \autoref{thm:height2} we characterize the lattices of stable matchings of graphs meeting the sufficient condition of \autoref{thm:eq-deg-match}. In \autoref{thm:outerplanar-lattice} we characterize the lattices of stable matchings of subcubic outerplanar and subcubic series-parallel graphs, and in \autoref{thm:planar-lattice} we characterize the lattices of stable matchings of planar graphs.

\item In \autoref{sec:hardness} we study the computational complexity of recognizing graphs of stably matchable pairs. We prove in \autoref{thm:npc} that it is NP-complete to determine whether a given graph can be realized as a graph of stably matchable pairs, even when the graph is planar and subcubic.

\item In \autoref{sec:algorithms} we study algorithms for recognizing the graphs of stably matchable pairs, more efficiently than a brute-force search through all preference systems on the given graph. We provide algorithms that can test whether an arbitrary graph is a graph of stably matchable pairs, either in time singly exponential in the number of edges (\autoref{thm:single-exponential}), or in time that is fixed-parameter tractable in the carving width of the given graph (\autoref{thm:cw-alg}).
\end{itemize}

\subsection{Related work}

The graphs of stably matchable pairs for random instances with uniformly-random preferences have been investigated by Knuth, Motwani, and Patel, who proved that with $n$ elements on each side the number of edges in these graphs is $\Theta(n\log n)$ with high probability~\cite{KnuMotPit-RSA-90}, and by Pittel, Shepp, and Veklerov, who bounded the numbers of vertices with degree~$d$ in these graphs, for constant~$d$, with high probability~\cite{PitSheVek-SIDMA-07}. In contrast to these results, random instances for which the elements on one side have preference lists of bounded size produce graphs in which almost all vertices have degree one~\cite{ImmMah-STOC-05}.

The part of our work on the lattices of stable matchings, and the connections between lattice structure and graph structure, is based on the work of Gusfield, Irving, Leather, and Saks on the distributive lattices derived from stable matching instances. As they proved, every finite distributive lattice can be derived from a stable matching instance in this way. There is a polynomial time algorithm for reversing the transformation from matching instances to lattices, and constructing a stable matching instance that has a given distributive lattice as its lattice of instances~\cite{GusIrvLea-JCTA-87}.

\section{Preliminaries}

We summarize below known and standard results on the structure of the system of stable matchings of a given stable matching instance. For convenience we will follow the convention that the elements being matched are students $s_i$ and residencies $r_j$. For book-length surveys of the stable matching problem see \cite{GusIrv-89,Knu-97,Man-13}.
\begin{itemize}
\item By the rural hospitals theorem, every stable matching has the same matched elements. Therefore, the elements of the given instance may be partitioned into isolated elements (those that cannot participate in any stable matching), uniquely matched elements (elements that are always matched to the same other element), and non-uniquely matched elements (always matched, but with more than one possible match).
\item The stable matchings may be partially ordered, with $M\le M'$ in this order when every student prefers $M$ to $M'$ (or has the same match in both and is indifferent), or equivalently when every residency prefers $M'$ to $M$ or is indifferent. This partial ordering forms a distributive lattice, in which the join (least upper bound) $M\vee M'$ of two matchings $M$ and $M'$ is obtained by giving every residency its more-preferred match from $M$ and $M'$ and giving every student their least-preferred match from $M$ and $M'$. Symmetrically, the meet (greatest lower bound) $M\wedge M'$ of two matchings is obtained by giving every residency its least-preferred match and giving every student their most-preferred match. The join and meet operations are associative and can be extended in the obvious way from pairs of stable matchings to arbitrary sets of stable matchings.
\item The top element $\top$ of the lattice (the join of all stable matchings) has the property that every residency prefers $\top$  to every other stable matching (or is indifferent to the comparison), and every student prefers any other stable matching to $\top$ (or is indifferent to the comparison). However, this matching is not necessarily obtained by giving every residency its first choice, as this might not result in a matching. If a residency is not assigned its first choice in $\top$, then that first choice is not available to it in any stable matching. Symmetric statements apply to the first choices of students and the bottom element $\bot$ of the lattice.
\item If a pair $r_i$ and $s_j$ cannot be matched in any stable matching of a given instance, then the system of all stable matchings is unchanged by changing the preferences so that both $r_i$ and  $s_j$ prefer being unmatched to being matched to each other. Therefore, any stable matching instance can be converted to an equivalent instance with partial preferences that only include pairs that can be matched. Conversely, if an instance with partial preferences has the property that its stable matchings include a match for every element, then it is equivalent to an instance with full preferences in which the additional pairings added to make the preference system full are ranked below all of the original pairings.
\item If two stable matchings $M$ and $M'$ are adjacent in the distributive lattice of matchings (meaning that they are comparable, say as $M\le M'$, and there is no third matching $M''$ between them in the partial order with $M\le M''\le M'$) then their symmetric difference is a \emph{rotation}, a single cycle of matched pairs that alternates between pairs of $M$ and of $M'$.
\item By Birkhoff's representation theorem for distributive lattices~\cite{Bir-DMJ-37}, the lattice of stable matchings can be recovered uniquely from the partial order of \emph{join-irreducible} stable matchings, the matchings $M$ that cannot be formed as the join of any subset of stable matchings that does not include $M$ itself. According to this theorem, the stable matchings correspond one-for-one with the \emph{lower sets} of this partial order, subsets of join-irreducible stable matchings for which, if $M'$ belongs to the subset and $M\le M'$, then $M$ also belongs to the subset. The stable matching corresponding to a lower set is its join, and the lower set corresponding to a stable matching $M$ is the lower set of join-irreducible stable matchings that are below $M$ in the lattice. The top stable matching $\top$ corresponds to the lower set of all join-irreducible stable matchings, and the bottom stable matching $\bot$ corresponds to the empty lower set. It is possible for $\top$ to be join-irreducible, but $\bot$ cannot be, as it is the join of the empty set, which does not contain~$\bot$.
\item For each join-irreducible stable matching $M'$, there is a unique stable matching $M\le M'$ that is adjacent to $M'$: $M$ is the join of all stable matchings below $M'$. We may associate $M'$ with the rotation formed by $M$ and $M'$. This association is one-to-one: every rotation comes in this way from a unique join-irreducible stable matching. Therefore, the rotations can be partially ordered, inheriting the partial order of the corresponding join-irreducible stable matchings. Within each rotation, the edges alternate between \emph{upper edges} (pairs in $M'$) and \emph{lower edges} (pairs in $M$); the same upper-lower relation on the edges remains valid for every pair of stable matchings that differ by the same rotation.
\item The collection of all rotations and their partial order can be constructed from a given system of preferences in polynomial time~\cite{Gus-SICOMP-87}.
\item Define an \emph{extended rotation} to be either a rotation or one of the two stable matchings $\top$ and $\bot$. Consider every pair in $\top$ to be a lower edge in $\top$, and consider every pair in $\bot$ to be an upper edge in $\bot$. Then for every pair that participates in at least one stable matching, there is a unique extended rotation in which it is a lower edge, and a unique extended rotation for which it is an upper edge.
\item The stable matching corresponding to a lower set $L$ in the partial order of rotations consists of the pairs $(s_i,r_j)$ that are upper edges in one of the extended rotations of $L\cup\{\bot\}$ but are not lower edges in any of these extended rotations.
\end{itemize}

Based on these results, we define a \emph{rotation system} for given sets of students and residencies to be a partially ordered collection $C$ of cycles that alternate between students and residencies, together with two matchings $\top$ and $\bot$ from students to residencies, such that each pair $(s_i,r_j)$ belongs to either zero or two members of $C\cup\{\top,\bot\}$, and such that each cycle alternates between upper edges (belonging to $\top$ or to an element above the cycle in the partial order) and lower edges (belonging to $\bot$ or to an element below the cycle in the partial order. We call the elements of $C$ \emph{rotations} and the elements of $C\cup\{\top,\bot\}$ \emph{extended rotations}. The results above imply that every stable matching instance can be described in polynomial time by a rotation system, and that the graph of stably matchable pairs in the instance equals the graph of pairs participating in extended rotations.

Conversely, we have:

\begin{lemma}
\label{lem:rot-is-stable}
Every rotation system describes the collection of stable matchings for a system of partial preferences. If in addition the rotation system has the property that there are either no unmatched students, or no unmatched residencies, then it describes the collection of stable matchings for a system of complete preferences.
\end{lemma}

\begin{proof}
A rotation system allows the edges incident with any one student or residency to be ordered, by the partial order of the rotations in which those edges are upper edges. Each incident edge that is not in $\top$ is the lower edge of at least one rotation, and has the upper edge of the same rotation above it; similarly, each incident edge that is not in $\bot$ is the upper edge of at least one rotation, and has the lower edge of the same rotation below it. Therefore, this partial ordering on incident edges forms a collection of ordered sequences in which there is a unique bottom element and a unique top element, which can only happen when there is a single chain forming a total order on these incident edges.
Based on this ordering of incident edges, define a system of partial preferences, in which the preference ordering for each residency is the same as its ordering of incident edges, and in which the preference ordering for each student is the reverse of its ordering of incident edges.

The sets of pairs coming from a lower set $L$ of rotations, as the pairs that are upper in an extended rotation of $L\cup\{\bot\}$ but not lower in any of these extended rotations, are all matchings, and we claim that they are all stable for these preferences. This follows by induction on the size of $L$. If $L$ is empty, the resulting set of pairs is $\bot$ itself, and it is stable because every student gets their first choice. Otherwise, some rotation $R$ in $L$ is maximal among the partially ordered elements of $L$, and by induction the matching obtained by removing $R$ from $L$ is stable. Then changing this matching by its symmetric difference with $R$ improves the preference of every residency whose matching is changed, but worsens the preference of every student whose matching is changed. This cannot introduce any new unstable pair $(s_i,r_j)$, because if $s_i$ prefers this pair, then it would already have the same preference in the stable matching prior to symmetric differencing with $R$, implying that $r_j$ did not prefer this pair and still does not prefer it.

Conversely, we claim that every stable matching for these preferences is generated in this way. If a matching $M$ is not generated in this way, but covers the same elements as $\top$ and $\bot$, then consider the set $R$ of rotations of the rotation system that contain an edge of $M$ as one of their upper edges, or are below one of these rotations. This is a lower set, and generates a stable matching $M'\ne M$. Because $M$ and $M'$ have the same number of edges, it follows that there is a cycle $C$ that is maximal in the partial order on $R$, in which at least one upper edge $(s_i,r_j)$  of $C$ is included in $M$ but not all are. However, for such an edge, $s_i$ prefers its other neighbor $r_j'$ (the adjacent lower edge in $C$). In order to prevent $(s_i,r_j')$ from being an unstable pair for $M$, $r_j'$ must be matched to an element it prefers to $s_i$, and this match can only be its other neighbor $s_i'$ in $C$. Any other match for $r_j'$ would either belong to an extended rotation lower than $C$ in the partial order, which $r_j'$ would not prefer to $s_i'$, or to an extended rotation higher than $C$, contradicting the choice of $C$ as being maximal in $R$.
It follows that the upper edge $(s_i',r_j')$ of $C$ also belongs to $M$. But repeating this argument step-by-step around $C$ shows that all upper edges of $C$ belong to $M$, contradicting the choice of $C$ as being a rotation whose upper edges are not all included in $M$. This contradiction establishes our claim that the given set of preferences does not generate any extra stable matchings beyond the ones coming from the given rotation system.

To obtain a full preference system from this partial preference system, when there are no unmatched students or no unmatched residencies, rank every added pairing below the pairings in the partial preference system. For the case that there are no unmatched students, $\bot$ is stable because it gives every student their first choice, establishing the base case of the induction. Every stable matching must have at least one matched pair belonging to the rotation system, because otherwise every pair in $\bot$ would be an unstable pair. The remaining steps of the proof, showing by induction that every matching from the rotation system is stable and showing by contradiction that no other stable matching can use a matched pair from the rotation system,  go through as before. The case that there are no unmatched residencies is symmetric and follows by reversing the roles of students and residencies (also reversing the comparisons in the partial order of the rotation system).
\end{proof}

For completeness we also include:

\begin{observation}
\label{obs:necessary}
Every graph of stably matchable pairs has the property that every connected component of more than two vertices is matching-covered and has a 2-factor (a 2-regular subgraph covering all of its vertices).
\end{observation}

\begin{proof}
Being matching-covered (within any connected component) follows from the definition, that the graph is covered by the stable matchings of some stable matching instance, and from the rural hospitals theorem, according to which all stable matchings cover the same set of vertices. The 2-factor within any component can be found as the disjoint union of the top and bottom matchings within the component, which (for a component of more than two vertices) must be disjoint.
\end{proof}

\begin{figure}[t]
\centering\includegraphics[scale=0.45]{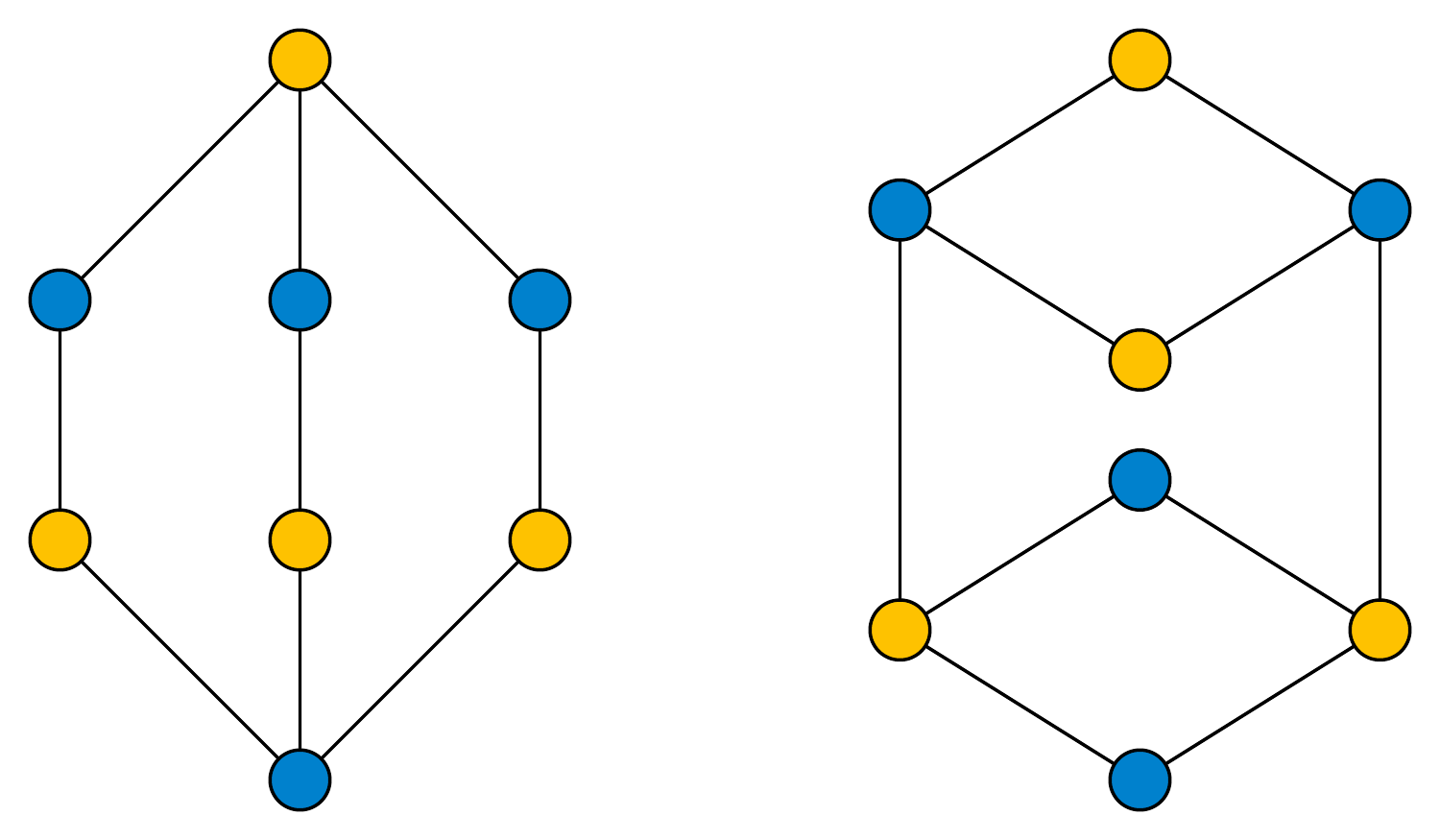}
\caption{Two balanced bipartite graphs that are not graphs of stably matchable pairs. The left graph is matching-covered but has no 2-factor; the right graph has a 2-factor (the union of the two rhombi through its top and bottom subsets of four vertices) but is not matching-covered.}
\label{fig:unmatchable}
\end{figure}

\autoref{fig:unmatchable} shows that the two conditions in \autoref{obs:necessary} are independent: both are biconnected subcubic balanced bipartite graphs, one is matching-covered but has no 2-factor, and the other has a 2-factor but is not matching-covered. Both conditions can be tested in polynomial time. With this foundational material established, we can turn to the example from \autoref{fig:k33-e}, which is matching-covered and has a 2-factor, but is still not a graph of stably matchable pairs.

\begin{observation}
\label{obs:k33-e}
The graph $K_{3,3}-e$ depicted in \autoref{fig:k33-e} is not the graph of stably matchable pairs of any stable matching instance.
\end{observation}

\begin{proof}
The four perfect matchings of this graph (shown on the right side of the figure) each use only one of the four central edges of the graph, so all four are needed to cover the whole graph. Because a single rotation can include only two of these edges (one as an upper edge and the other as a lower edge of the rotation), each rotation must pass through one of the degree-two vertices. Each degree-two vertex can be in only one rotation, so there are at most two rotations. For two rotations to generate four perfect matchings, they must be incomparable in the partial order of rotations, because a partial order with two comparable elements has only three lower sets. Rotations that are incomparable must be edge-disjoint, but this graph does not have two edge-disjoint cycles.
\end{proof}

\section{Which graphs are representable?}
\label{sec:representable}

Our later NP-completeness result (\autoref{thm:npc}) strongly suggests that a simple and exact characterization of the graphs of stably matchable pairs, in general, does not exist. Therefore, in this section, we seek special classes of graphs for which a characterization can be found, as well as partial characterizations in terms of conditions that are necessary or sufficient but not both.

\subsection{Regular graphs and induced subgraphs}
\label{sec:regular}

\begin{figure}[t]
\centering\includegraphics[width=0.5\textwidth]{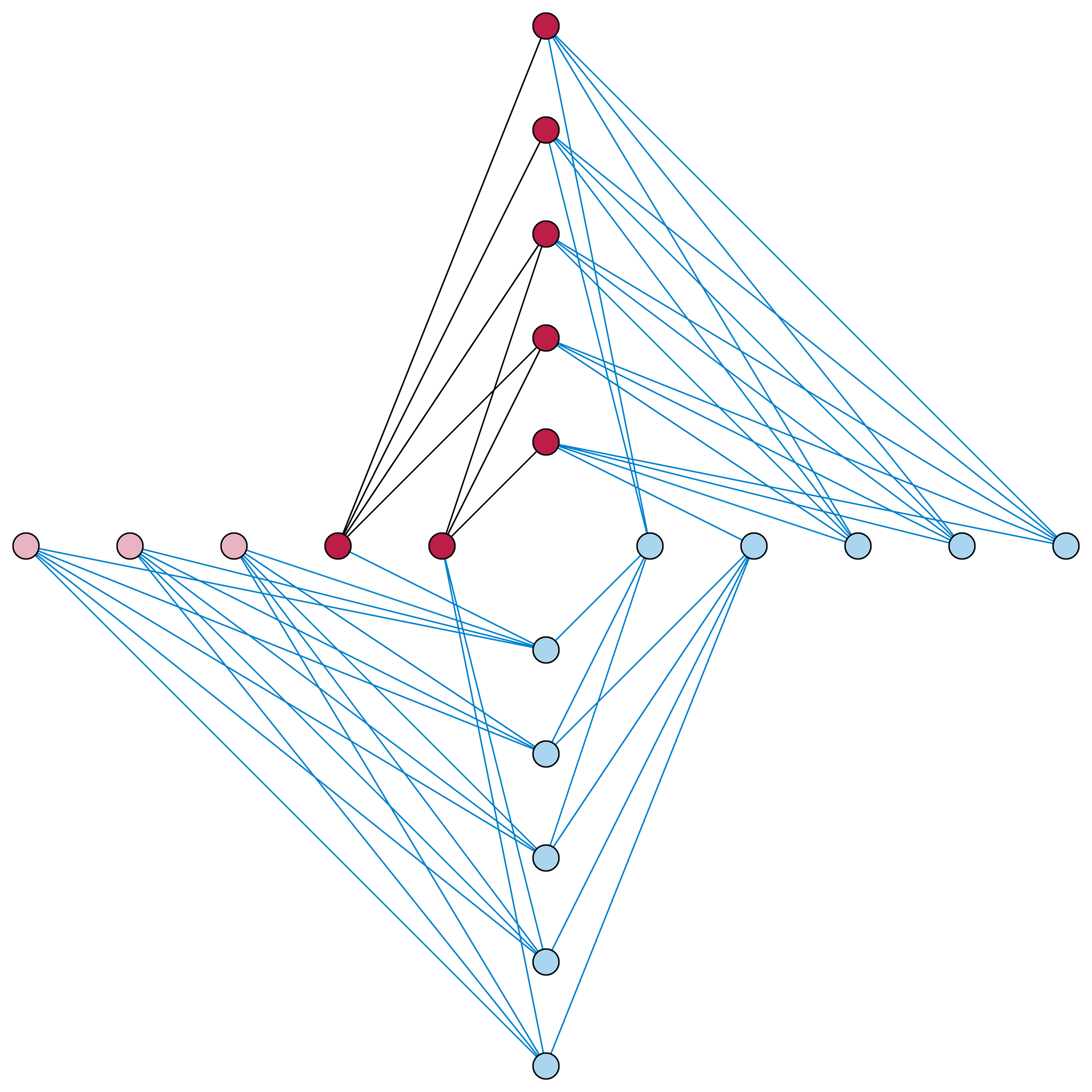}
\caption{Embedding an arbitrary bipartite graph $G$ (dark red vertices and black edges) into a regular bipartite graph (\autoref{lem:regularize}). The light pink vertices are the isolated vertices added to $G$ to make it balanced, and the light blue vertices are $s'_i$ and $r'_i$ from the construction of \autoref{lem:regularize}.}
\label{fig:regularize}
\end{figure}

In this section we show that the graphs of stably matchable pairs have no hidden structure, in the sense that there are no forbidden induced subgraphs for these graphs. We begin with the following observation:

\begin{observation}
\label{obs:regular}
Every regular bipartite graph is a graph of stably matchable pairs for some stable matching instance.
\end{observation}

\begin{proof}
By a theorem of D\'enes K\H{o}nig, every regular bipartite graph can be decomposed into disjoint perfect matchings~\cite{Kon-MA-16}.
Arbitrarily order these matchings, and construct a partial preference system in which one side of the bipartition (the students) prioritizes each potential pairing of the graph according to this ordering, and the other side of the bipartition (the residencies) prioritizes the pairings according to the reverse ordering. With these priorities, every matching of the decomposition (as well as, potentially, some other matchings using the same edges) is stable.
\end{proof}

For instance, this applies to the graph of the cube, to every balanced complete bipartite graph, to every Levi graph of a finite projective plane (such as the Heawood graph, the Levi graph of the Fano plane), and to every Levi graph of a projective configuration of equally many points and lines with each point on the same number of lines and each line containing the same number of points (such as the M\"obius--Kantor graph, Pappus graph, and Desargues graph, Levi graphs of configurations with 8, 9, and 10 points and lines respectively).

\begin{lemma}
\label{lem:regularize}
Every bipartite graph is an induced subgraph of a regular bipartite graph.
\end{lemma}

\begin{proof}
Let $G$ be an arbitrary bipartite graph. Add isolated vertices to one side of the bipartition of $G$, if necessary, to make the resulting graph become balanced. Let $s_i$ and $r_j$ be the vertices on the two sides of the bipartition of the resulting balanced bipartite graph. Construct a graph $G'$ with twice as many vertices, by making a copy $s'_i$ or $r'_j$ of each vertex $s_i$ and $r_j$. For each edge $s_ir_j$ in $G$, make two edges in $G'$, $s_ir_j$ and $s'_ir'_j$, and for each pair $(s_i,r_j)$ that is not an edge in $G$, make two edges in $G'$, $s_ir'_j$ and $s'_ir_j$, as shown in \autoref{fig:regularize}. The result is a regular balanced bipartite graph in which the vertices $s_i$ and $r_j$ induce a copy of $G$, and in which every vertex on one side of the bipartition is adjacent to exactly half of the vertices on the other side of the bipartition.
\end{proof}

\begin{corollary}
Every bipartite graph is an induced subgraph of a graph of stably matchable pairs.
\end{corollary}

\begin{proof}
It is an induced subgraph of a regular bipartite graph by \autoref{lem:regularize}, and this regular bipartite graph is a graph of stably matchable pairs by \autoref{obs:regular}.
\end{proof}

\subsection{Subcubic graphs}

A \emph{subcubic graph} is a graph such as our example $K_{3,3}-e$ in which every vertex has degree at most three. It all degrees are exactly three, it is a \emph{cubic graph}. In this section we examine more closely the structure of subcubic graphs of stably matchable pairs, finding both necessary conditions and (different) sufficient conditions for a subcubic graph to be representable in this way. Our next result provides necessary and (different) sufficient conditions for being a graph of stably matchable pairs, based on matchings of equal-degree vertices. It is convenient to introduce the notation $G[i]$, where $G$ is a graph and $i$ is a natural number, meaning the subgraph of $G$ induced by its degree-$i$ vertices.

\begin{theorem}
\label{thm:eq-deg-match}
If $G$ is a bipartite subcubic graph, then a necessary condition for $G$ to be a graph of stably matchable pairs is that $G[1]$ and $G[3]$ both have perfect matchings. It is sufficient that, as well, $G[2]$ has a perfect matching.
\end{theorem}

\begin{proof}
If $G$ is to be a graph of stably matchable pairs, it is necessary that every non-isolated vertex of $G$ be matched in every stable matching. If $v$ is a degree-one vertex, its neighbor can have no other matches, so its neighbor must also be degree-one. Therefore, it is necessary for $G[1]$ to have a perfect matching.

If $G$ is a subcubic graph of stably matchable pairs of some stable-matching instance, and $v$ is a vertex of degree two or three in $G$, then one of the edges incident to $v$ must belong to the top matching $\top$ of the lattice of stable matchings, and a different one of the edges incident to $v$ must belong to the bottom matching $\bot$. Each degree-three vertex is incident to a unique edge that belongs neither to $\top$ nor $\bot$, so these edges form a perfect matching of $G[3]$.
 
If $G[1]$, $G[2]$, and $G[3]$ all have perfect matchings (\autoref{fig:eq-deg-match}, far and middle left), then the subgraph formed by removing from $G$ the matched edges in $G[3]$ consists of isolated vertices and edges, and disjoint cycles of degree-two and degree-three vertices (\autoref{fig:eq-deg-match}, middle right). Form two matchings $\top$ and $\bot$ consisting of the matched edges in $G[1]$ and alternating subsets of these disjoint cycles, choosing arbitrarily within each cycle how to alternate between edges of $\top$ (which we call \emph{top edges}) and edges of $\bot$ (which we call \emph{bottom edges}). Call the edges in the matching of $G[3]$ \emph{middle edges}, as they are neither in $\top$ nor in $\bot$. Additionally, for a top edge or a bottom edge, we say that it is \emph{matched} if it is part of the perfect matching of $G[2]$, and unmatched otherwise.
Once this choice has been made, form a system of cycles in $G$ of two types (\autoref{fig:eq-deg-match}, far right):
\begin{itemize}
\item \emph{Upper cycles} alternate between unmatched top edges and either middle edges (at degree-3 vertices of $G$) or matched bottom edges (at degree-2 vertices of $G$).
\item \emph{Lower cycles} alternate between unmatched bottom edges and either middle edges or matched top edges.
\end{itemize}
These cycles are simply the connected components of a 2-regular graph formed by splitting each degree-three vertex and each middle edge into two copies, an upper copy and a lower copy, with the incident top edge going to the upper copy and the incident bottom edge going to the lower copy. In particular, they are simple cycles, with each top or bottom edge appearing in exactly one of these cycles and each middle edge appearing in one upper cycle and one lower cycle.

\begin{figure}[t]
\centering\includegraphics[width=0.8\textwidth]{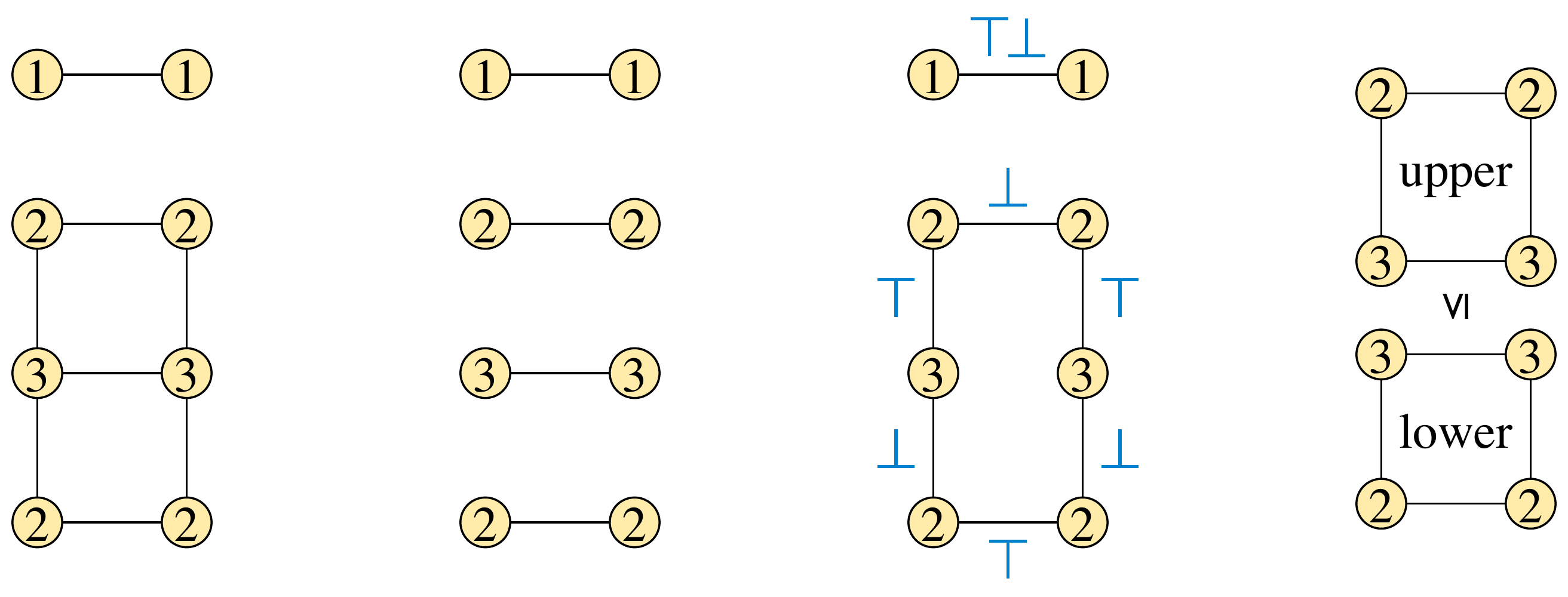}
\caption{Illustration for the sufficient case of \autoref{thm:eq-deg-match}. Far left: A subcubic graph $G$, with vertices labeled by their degrees. Middle left: simultaneous perfect matching of the subgraphs $G[1]$, $G[2]$, and $G[3]$ of equal-degree vertices. Middle right: assignment of the edges that are not matched in $G[3]$ to the matchings $\top$ and $\bot$, alternating around the remaining cycle of degree-2 and degree-3 vertices. Far right: the upper cycle and lower cycle derived from the choice of $\top$ and $\bot$, and the rotation structure formed by these cycles.}
\label{fig:eq-deg-match}
\end{figure}

Partially order these cycles by making two cycles $C$ and $C'$ comparable when they share one or more middle edges, with $C\le C'$ when $C$ is the lower cycle and $C'$ is the upper cycle of the two, and by making all other pairs of cycles incomparable. Then this system of partially ordered cycles, together with the matchings $\top$ and $\bot$, form a rotation system for $G$, and the result that $G$ is a graph of stably matchable pairs follows from \autoref{lem:rot-is-stable}.
\end{proof}

\subsection{Outerplanar graphs}

Outerplanar graphs can be tested for being a graph of stably-matching pairs very easily, by checking a simple necessary and sufficient condition.

\begin{lemma}
\label{lem:articulation}
No graph of stably-matching pairs can have an articulation vertex.
\end{lemma}

\begin{proof}
A graph with an articulation vertex $v$ cannot be matching-covered, and so by \autoref{obs:necessary} cannot be a graph of stably-matching pairs. For, if any component formed from the graph by the removal of $G$ has an even number of vertices, the edge or edges connecting that component to $v$ cannot be included in any perfect matching. And if all of these components have an odd number of vertices, only one of these components can have an edge to $v$ in a perfect matching, and the rest of these components cannot be perfectly matched.
\end{proof}

\begin{theorem}
\label{thm:outerplanar}
A bipartite outerplanar graph $G$ is a graph of stably-matching pairs if and only if $G$ has no articulation vertex.
\end{theorem}

\begin{proof}
One direction is \autoref{lem:articulation}. In the other direction, let $G$ be a bipartite outerplanar graph, drawn in an outerplanar way (with all vertices belonging to the unbounded face). We prove by induction on the number of bounded faces that $G$ is a graph of stably matchable pairs for a rotation system in which the outer face of $G$ is the union of the top and bottom matchings, and (when $G$ has more than one bounded face) all bounded faces are rotations.

 As a base case, if $G$ is a cycle graph, it has a rotation structure consisting of two disjoint perfect matchings (the top and bottom matchings of the rotation structure) and an empty set of rotations. As a second base case, if $G$ consists of two cycles sharing an edge, it has a rotation structure where these two cycles are both rotations, alternating between top edges and bottom edges except at their shared edge. One of these two rotations (the upper one in the partial order of the two rotations) has top edges adjacent to the shared edge, and the other one has bottom edges adjacent to the shared edge.
 
In the non-base case, the weak dual of $G$ (a graph with a vertex for each bounded face and an edge between bounded faces that share an edge in $G$) is a tree with at least one leaf $f$, a bounded face of $G$ that shares only one edge $e$ with the rest of the graph. By induction, the outerplanar graph $G'$ formed from $G$ by removing $f$ (but keeping the shared edge $e$ in place) has a rotation structure in which the outer face of $G'$ (including $e$) alternates between top and bottom edges, and each face is a rotation. We construct a rotation structure for $G$ by adding $f$ as a rotation to this structure, alternating between top and bottom edges. If $e$ was a top edge in $G'$, we place $f$ above the neighboring rotation in the partial order, and choose the alternation of top and bottom edges of $f$ to have top edges adjacent to $e$. Otherwise, we place $f$ below the neighboring rotation, and choose the alternation of its edges with bottom edges adjacent to $e$. Both cases preserve all the defining properties of a rotation system.
\end{proof}

Although more general bipartite outerplanar graphs may be unbalanced, the bipartite outerplanar graphs with no articulation vertex are automatically balanced. This follows as a corollary from \autoref{thm:outerplanar} but can also be proven more directly by a similar induction.

\subsection{Rectangular grids}

We define an $a \times b$ \emph{grid graph} to be the graph of unit distances in the integer points $(i,j)$ with $0\le i<a$ and $0\le j<b$. That is, it consists of a rectangular array of vertices, with $a$ vertices along one side of the rectangle and $b$ along the other, and with vertices adjacent to their nearest neighbors in the directions parallel to the rectangle sides. When $ab$ is even, the resulting graph is automatically balanced.
\begin{figure}[t]
\includegraphics[width=\textwidth]{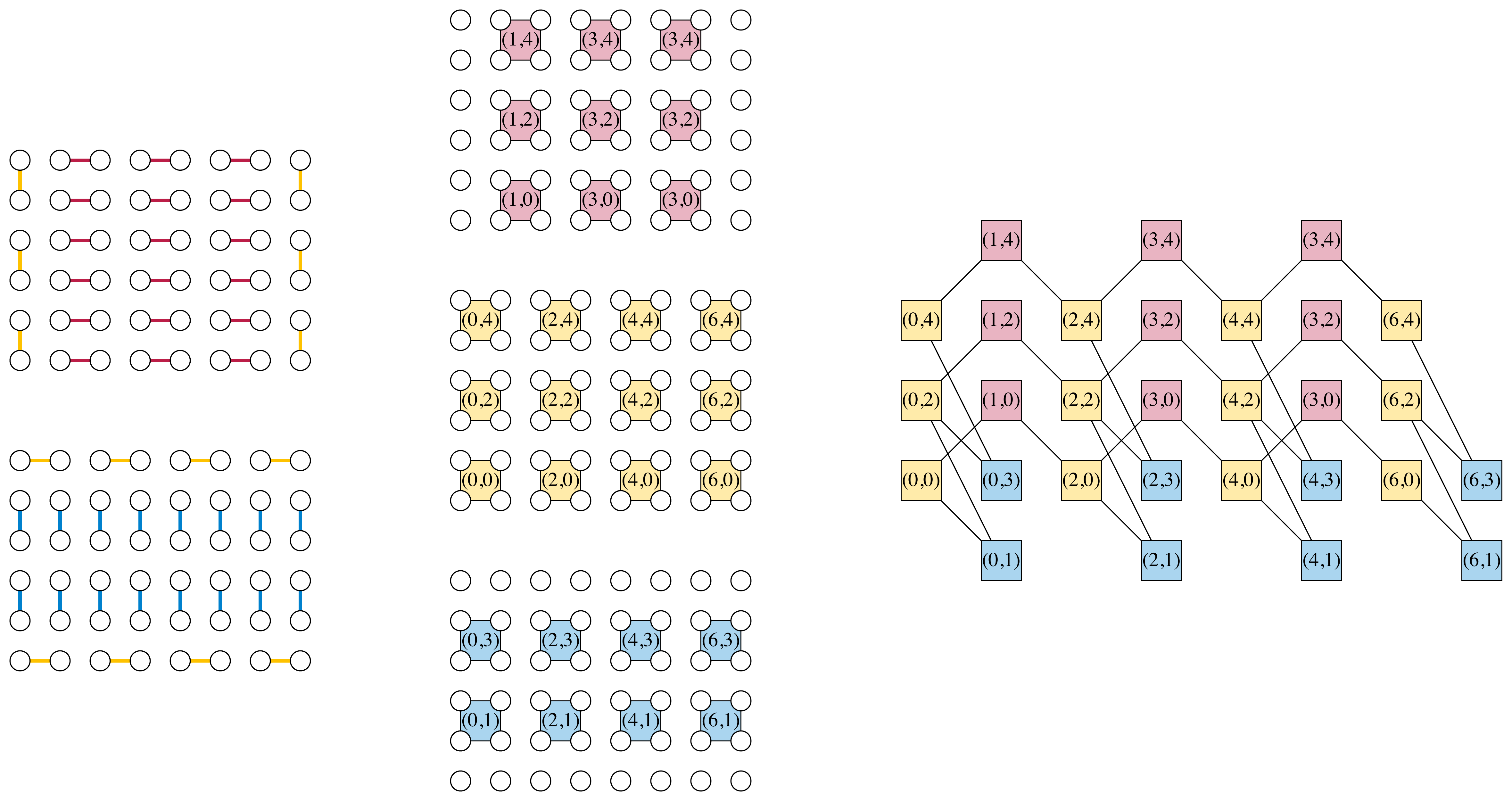}
\caption{A rotation system realizing an $8\times 6$ grid as a graph of stably matchable pairs, following the construction of \autoref{thm:grid}. Left: the top matching (upper left, red and yellow edges) and the bottom matching (lower left, blue and yellow edges) of the rotation system. Center: the squares used as rotations in the rotation system, and their grid coordinates.
Right: the partial order on rotations.}
\label{fig:6x8}
\end{figure}

\begin{theorem}
\label{thm:grid}
The $a \times b$ grid graph $(a,b>1)$ can be realized as a graph of stably matchable pairs if and only if $ab$ is even and $\min\{a,b\}\ne 3$.
\end{theorem}

\begin{proof}
We consider the following cases:
\begin{itemize}
\item If $ab$ is odd, the number of vertices is odd and no perfect matching can exist.

\item If $\min\{a,b\}=2$, the graph is outerplanar and realizability follows from \autoref{thm:outerplanar}.

\item If $\min\{a,b\}=3$, suppose for a contradiction that the grid has a rotation system. Consider the block $B$ of six squares at one end of the grid, separated by three edges from the rest of the grid. In any rotation system, the bottom matching includes an even number of these separating edges. A rotation that does not include these edges does not change the number of separating edges that are used. A rotation that uses the middle of these three edges either adds two to the number of separating edges that are used, or removes two of these edges from use, because (by the parity of the distance between the endpoints of the two separating edges that it uses) either both separating edges must be upper in the rotation or both must be lower. However, in any sequence of stable matchings progressing upward through the lattice of stable matchings from its bottom to its top, each of the three separating edges must be introduced exactly once, so there must be a rotation that introduces only one new edge. Such a rotation $R$ must necessarily use the outer two of the three edges separating $B$ from the rest of the graph; again, by the parity of the distance between the endpoints of these outer separating edges, one of them would be an upper edge of any such rotation and the other of them would be a lower edge. Within $B$, the path of $R$ uses an odd number of vertices, so in order for all vertices of $B$ to be matched, the two perfect matchings that are connected by $R$ must both include the middle separating edge of $B$ as one of their matched edge. The remaining five vertices of $B$ induce a path, which must be the path followed by $R$. Because the two corner vertices of this path have degree two, the four edges of this path must all belong to $\top$ or $\bot$. But this means that the rotation $R$ passes through the degree-three vertex at the center of the path using only edges of $\bot$ or $\top$, making it impossible for any other rotation to incorporate the third edge at that vertex into the rotation system. This contradicts our assumption that a rotation system exists, proving that no realization as a graph of stably matchable pairs is possible in this case.

\begin{figure}[t]
\centering\includegraphics[width=0.9\textwidth]{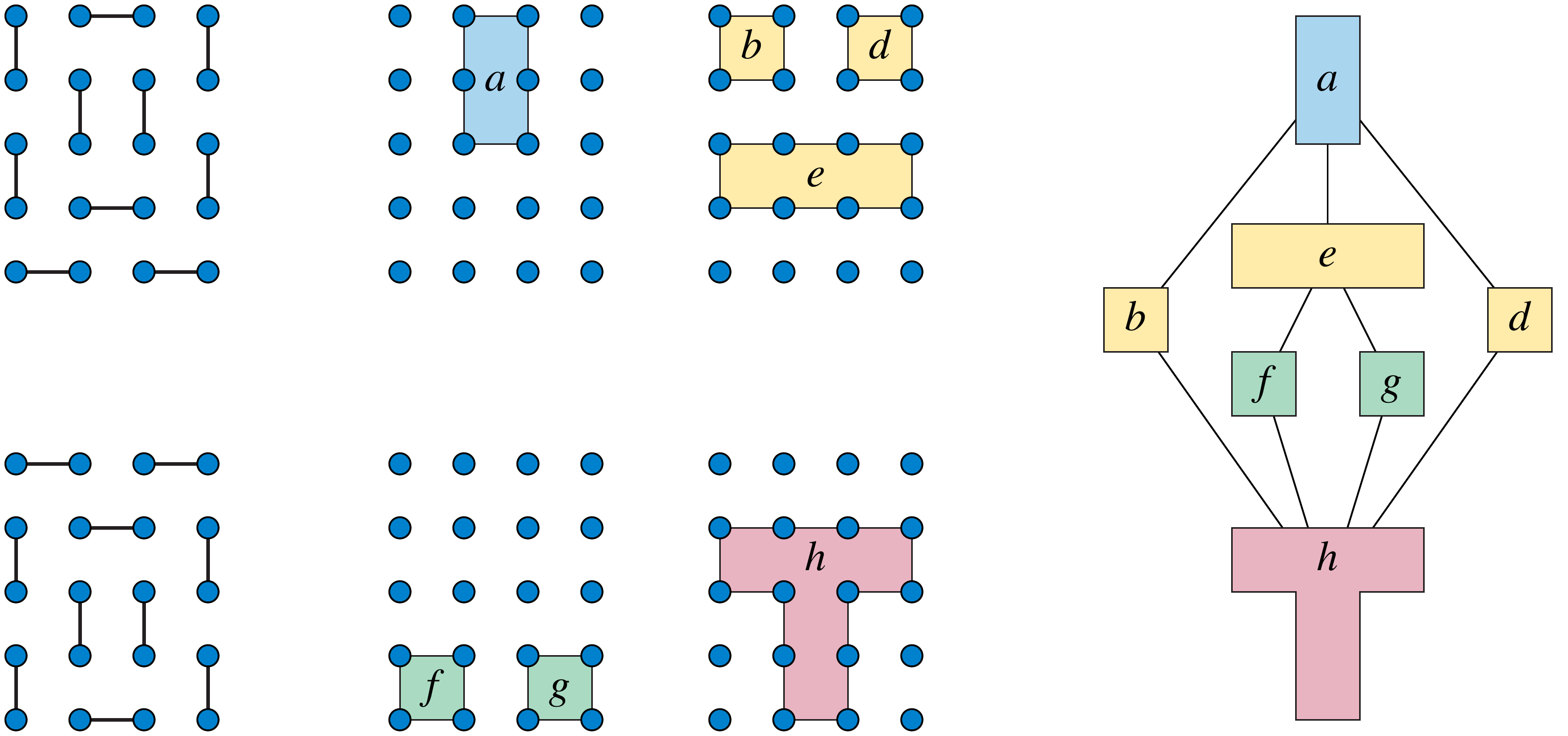}
\caption{A rotation system realizing a $4\times 5$ grid as a graph of stably matchable pairs.}
\label{fig:5x4}
\end{figure}

\begin{figure}[t]
\centering\includegraphics[width=0.9\textwidth]{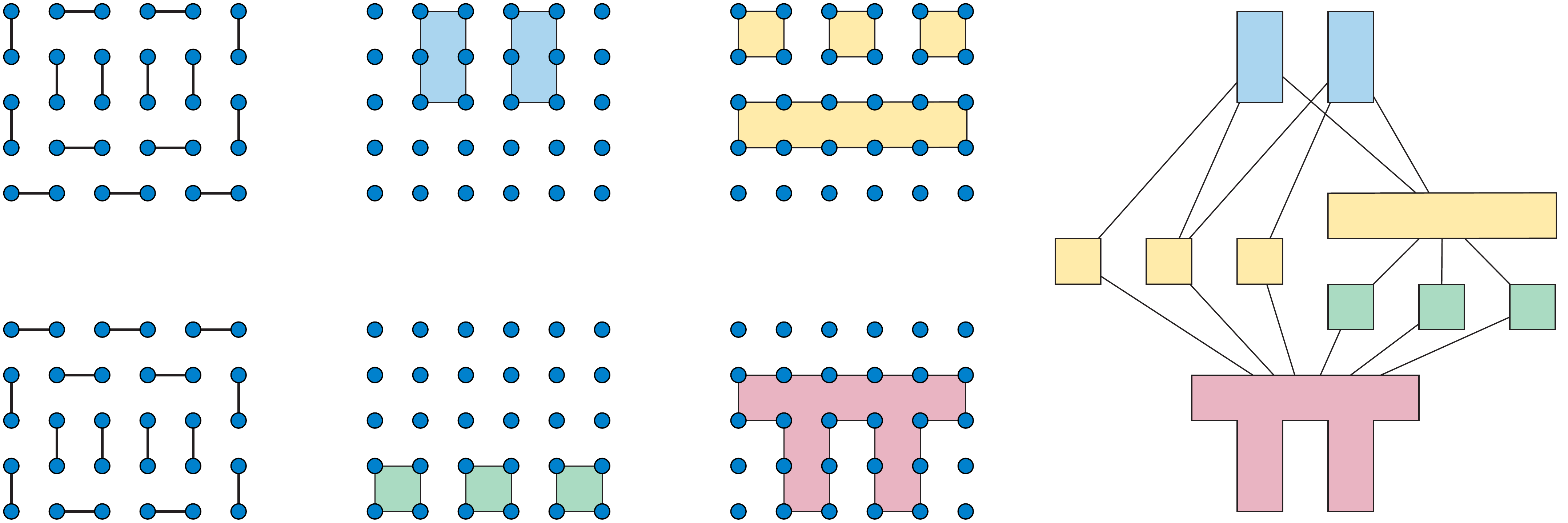}
\caption{A rotation system realizing a $6\times 5$ grid as a graph of stably matchable pairs. Repeating the same pattern of top and bottom matchings and rotations produces a rotation system for any $(2i)\times 5$ grid with $i\ge 2$.}
\label{fig:5x6}
\end{figure}

\item If $a$ and $b$ are both even, we form a rotation system using a subset of the unit squares of the grid as rotations, as illustrated in \autoref{fig:6x8}. If we associate a grid square with its lower left corner (the point of the square whose coordinates are smallest), then we choose the squares whose first coordinate is even and whose second coordinate as odd as the bottom elements of the partial order of rotations. We choose the squares with both coordinates even as intermediate elements in this partial order, and we choose the squares whose first coordinate is odd and whose second coordinate is even as the top elements of the partial order. We do not use squares with both coordinates odd as rotations. The top matching consists of horizontal edges of top squares, and vertical edges of intermediate squares that are not shared with top squares. The bottom matching consists of vertical edges of bottom squares, and horizontal edges of intermediate squares that are not shared with bottom squares.

\item In the remaining cases, the smallest odd dimension of the grid is at least five and the smallest even dimension is at least six. \autoref{fig:5x4} depicts a rotation system realizing the $4\times 5$ grid and \autoref{fig:5x6} shows how to extend this to all $(2i)\times 5$ grids for $i\ge 2$. Examination of each grid edge in the figures shows that they are all covered by two elements of the rotation system and top and bottom matching, one upper and one lower, and examination of each rotation shows that they all alternate between upper and lower edges, as required for a rotation system.

Note that, in each of the resulting rotation systems, the top matching includes a perfect matching of the vertices on the bottom rows of the rotation system. This allows these rotation systems to be extended by any even number $2i$ of rows below these bottom rows, by using the square rotations from the bottom $2i+1$ squares of the pattern used in the previous case for even-by-even grids, maintaining the partial order among these rotations and placing them in the partial order above all of the rotations used for the top five rows of the grid.
\end{itemize}
In all cases with $ab$ even, the grid is realizable when $\min\{a,b\}\ne 3$ and unrealizable otherwise.
\end{proof}

\subsection{Product graphs}

Given two graphs $G$ and $H$, their Cartesian product $G\mathop{\square} H$ is a graph whose vertices belong to the Cartesian product of sets $V(G)\times V(H)$ and whose edges connect two elements of $V(G)\times V(H)$ (two pairs of a vertex in $G$ and a vertex in $H$) when the vertices in $G$ are equal and the vertices in $H$ are adjacent, or vice versa. We may identify the edge set of the product graph with the disjoint union $\bigl(V(G)\times E(H)\bigr)\cup\bigl(E(G)\times V(H)\bigr)$. For instance, the grid graphs of the previous section are Cartesian products of path graphs. If both $G$ and $H$ are bipartite, then so is their product: if we two-color both $G$ and $H$, then one side of the bipartition of $G\mathop{\square} H$ consists of monochromatic pairs of a vertex from $G$ and a vertex from $H$, and the other side consists of bichromatic pairs.

\begin{theorem}
\label{thm:product}
If both $G$ and $H$ are realizable as graphs of stably matching pairs, then so is $G\mathop{\square} H$.
\end{theorem}

\begin{proof}
We assume without loss of generality that neither $G$ nor $H$ has isolated vertices, as any such vertices lead in $G\mathop{\square} H$ to disconnected copies of $G$ or $H$ that can be handled separately.

A preference system realizing $G\mathop{\square} H$ from separate realizations of $G$ and $H$ can be derived by giving the students preferences that prioritize edges coming from $G$ over edges coming from $H$, and giving the residencies preferences that prioritize edges coming from $H$ over edges coming from $G$. However it is easier to verify that this does not lead to additional undesired matched edges if we describe more explicitly a rotation system realizing $G\mathop{\square} H$, derived from rotation systems for $G$ and for $H$.

Let $\top_G$, $\top_H$, $\bot_G$, and $\bot_H$ denote the top and bottom matchings of the rotation systems for $G$ and $H$. We form a rotation system whose top matching is $\top=V(G)\times \top_H$ (using the top matching in each copy of $H$, as preferred by the residencies) and whose bottom matching is $\bot=\bot_G\times V(H)$. For each rotation $R$ in $G$ we form $|V(H)|$ copies of $R$ in the rotation structure for $G\mathop{\square} H$, one for each vertex in $H$, and similarly for each rotation in $H$ we form $|V(G)|$ copies in the rotation structure for $G\mathop{\square} H$, one for each vertex in $G$. These copies are partially ordered in the same way they are in the rotation systems of $G$ or $H$, with two copies of rotations being incomparable when they correspond to different vertices in $G$ or $H$.

Finally, we add to these rotations the cycles of the 2-regular graph $(\top_G\times V(H))\cup (V(G)\times\bot_H)$.
Each such cycle is partially ordered above the copies of rotations from $G$ with which it shares an edge, and below the copies of rotations from $H$ with which it shares an edge, causing it to alternate between upper and lower edges (as required in a rotation system) and providing each of the edges of this 2-regular graph with a second rotation that it belongs to (as also required in a rotation system).
\end{proof}

The examples of the grid graphs show that, although the realizability of both $G$ and $H$ is sufficient for the realizability of $G\mathop{\square} H$, it is not always necessary. The example of $K_2\mathop{\square}K_{1,3}$ in \autoref{fig:unmatchable} (left) shows that realizability of a single factor and bipartiteness of the other factor is not always sufficient.

\section{How does the graph constrain the lattice?}
\label{sec:lattice}

In this section we revisit the result of  Gusfield, Irving, Leather, and Saks that every finite distributive lattice is the lattice of stable matchings of a stable matching instance~\cite{GusIrvLea-JCTA-87}. We ask whether this remains true for natural subclasses of the graphs of stably matchable pairs, or whether restricting the class of graphs that we consider also restricts the classes of lattices that they can realize.

\subsection{Background and definitions}

Several of our characterizations will relate the lattices of stable matchings to the lattices of closures of a directed graph. In this context, a \emph{closure} is a subset of the vertices of a given directed graph that has no outgoing edges: there is no edge from a vertex in the closure to a vertex outside the closure~\cite{Pic-MS-76,Epp-TALG-18}. The closures of a graph, ordered by subsets, form a distributive lattice, which is the lattice of lower sets in a partial order on strongly connected components of the given graph, defined by setting $X\le Y$ whenever a vertex of component $X$ can be reached from a vertex of $Y$. A maximum-weight closure can be found in polynomial time, and maximum-weight closures on directed acyclic graphs of rotations can be used to find stable matchings that meet additional optimization criteria~\cite{IrvLeaGus-JACM-87}.

A directed graph may be obtained from an undirected graph by choosing a direction for each undirected edge, and a directed graph constructed in this way is called an \emph{orientation} of the undirected graph; for instance, \emph{oriented planar graphs} are the orientations of undirected planar graphs. A \emph{minor} of a graph is any graph that can be obtained from it by a (possibly empty) sequence of edge contractions, edge deletions, and vertex deletions.

A \emph{comparability} graph is a graph whose vertices are the elements of a partially ordered set and whose edges connect pairs of elements that are related in the partial order. An \emph{intersection graph} is a graph whose vertices are members of a family of sets and whose edges connect pairs of vertices whose sets have a nonempty intersection. A \emph{string graph} is an intersection graph of a family of curves in the plane~\cite{EhrEveTar-JCT-76}. (Sometimes a restriction is added that at most two curves intersect at any point; this does not affect the resulting class of graphs, as any family of curves can be perturbed to meet this restriction without adding or removing any intersection pairs.)

\subsection{Subcubic graphs}
Although the condition that every element have only three choices of matches may seem restrictive, subcubic graphs still have much of the general structure of arbitrary stable matching instances, as the following result demonstrates.

\begin{figure}[t]
\includegraphics[width=\textwidth]{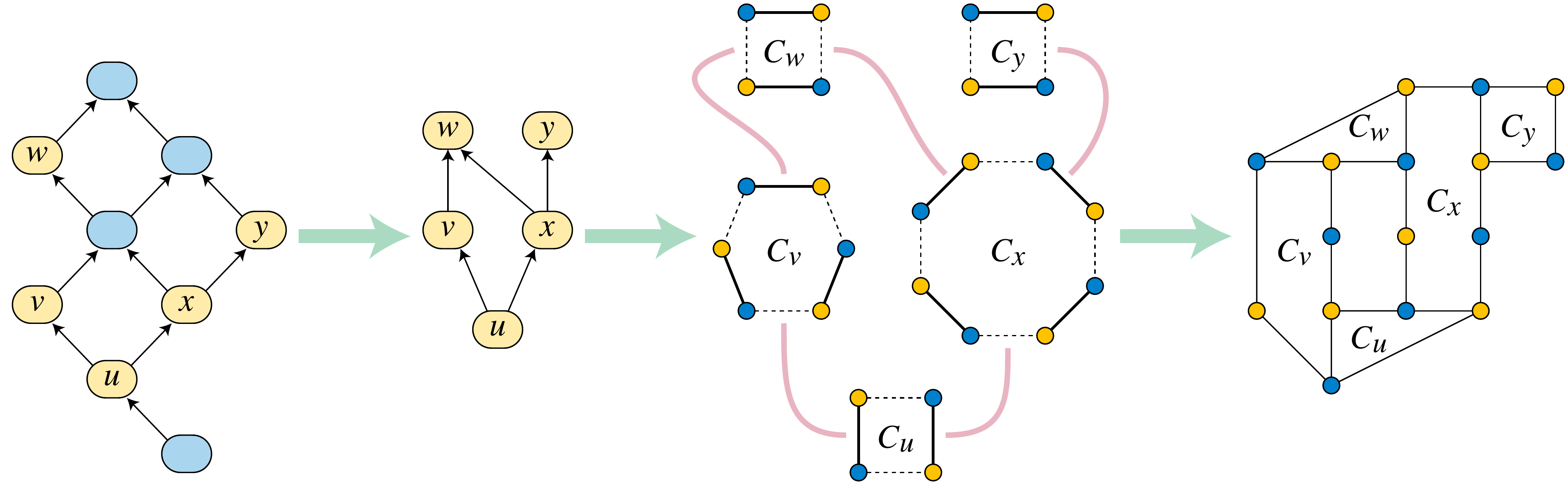}
\caption{Construction of \autoref{thm:dl2sub3} from an arbitrary finite distributive lattice (far left) through the Hasse diagram of its partial order of join-irreducibles (middle left) to an alternating cycle for each join-irreducible (middle right, with upper edges shown solid and lower edges dashed, and with light red curves showing pairs of cycle edges both associated with the same Hasse diagram edge), merged to form a subcubic graph of stably matchable pairs (far right).}
\label{fig:dl2sub3}
\end{figure}

\begin{theorem}
\label{thm:dl2sub3}
Every finite distributive lattice is isomorphic to the lattice of stable matchings of a stable matching instance whose graph of stably matchable pairs is subcubic.
\end{theorem}

\begin{proof}
From an arbitrary given finite distributive lattice, construct the partial order of its subset of join-irreducible elements, and let $H$ be the Hasse diagram of this partial order. That is, $H$ has a vertex for each join-irreducible element of the lattice, and an edge from the lower to the upper of each two comparable elements that have no other join-irreducible element between them.

For each vertex $v$ of $H$, construct a cycle graph $C_v$, of even length, with its edges labeled alternatingly as upper and lower and with its vertices labeled alternatingly as students and residencies. Each edge $u\rightarrow v$ in $H$ should be associated with one of the lower edges of $C_v$, and each edge $v\rightarrow w$ in $H$ should be associated with one of the upper edges of $C_v$, in such a way that the edges of $C_v$ that are associated with edges of $H$ in this way are not adjacent to each other in $C_v$. If $v$ has total degree $d_v$ in $H$, then $C_v$ can be constructed with length $4$ if $d_v=1$, or with length at most $2d_v+4$ if $d_v>1$.

Form a graph $G$ from the union of the cycles $C_v$ by, for each edge $u\rightarrow v$ in $H$, merging together the edges of $C_u$ and $C_v$ that are associated with $u\rightarrow v$ (at the same time merging the two student endpoints of these edges and the two residency endpoints of these edges). Then $G$ is subcubic, because each merged pair of endpoints produces a vertex of $G$ with degree three (one merged edge and its two neighbors in the two cycles from which it was merged) and each unmerged vertex of one of the cycles $C_v$ continues to have degree two.
Form also a rotation system in $G$, in which the rotations are the cycles $C_v$, partially ordered by the partial order for which $H$ is the Hasse diagram, in which $\top$ is the matching formed by the unassociated upper edges of cycles $C_v$, and in which $\bot$ is the matching formed by the unassociated lower edges of cycles $C_v$. Both $\top$ and $\bot$ are perfect matchings, because each merged or unmerged vertex in $G$ has exactly one neighbor in $\top$ and one in~$\bot$.

Because the subcubic graph $G$ has a rotation system, it follows from \autoref{lem:rot-is-stable} that there is a system of preferences for which this rotation system describes all of the stable matchings. Because both the lattice of stable matchings for these preferences and the initially given finite distributive lattice can be described as the lower sets of isomorphic partial orders, they are isomorphic lattices.
\end{proof}

The construction of \autoref{thm:dl2sub3} is illustrated in \autoref{fig:dl2sub3}.

Recall that in \autoref{thm:eq-deg-match} we proved that if a subcubic graph has a perfect matching where all matched edges connect equal-degree vertices, then it is a graph of stably matchable pairs. Call a graph with this property \emph{degree-matchable}. The \emph{height} of a partial order is the maximum number of elements in any chain, a totally-ordered subset of elements. A \emph{bipartite orientation} of an undirected bipartite graph is an orientation in which all edges are directed in the same way, from one side of the bipartition to the other. The closures of a bipartite orientation are the same as the lower sets of a partial order of height at most two, on the vertices of the bipartite graph, with $u\le v$ if there is an edge from $v$ to $u$.

\begin{theorem}
\label{thm:height2}
If a stable matching instance has a subcubic degree-matchable graph as its graph of stably matchable pairs,
then its lattice of stable matchings is isomorphic to the lattice of closures of a bipartite orientation, or equivalently the lattice of lower sets of a partial order of height at most two. Every lattice of closures of a bipartite orientation, or lattice of lower sets of a partial order of height at most two, can be realized as the lattice of stable matchings
of an instance whose graph of stably matchable pairs is subcubic and degree-matchable.
\end{theorem}

\begin{proof}
Let $G$ be a degree-matchable graph of stably matchable pairs. The degree-zero and degree-one vertices of $G$ are irrelevant for its lattices of stable matchings, so assume without loss of generality that there are none. Consider any rotation system for $G$; we wish to show that the lattice of stable matchings of this rotation system is isomorphic to the closures of a bipartite orientation, whose underlying undirected bipartite graph will be the comparability graph of the rotations. To see this, label the edges of $G$ as top, bottom, or middle, accordingly. Then each rotation must alternate between upper and middle edges, or between middle and lower edges, at the degree-3 vertices of $G$, and between upper and lower edges at the degree-2 vertices of $G$. Because $G$ is degree-matchable, every part of a rotation that passes through degree-2 vertices starts and ends on an edge with the same label, and does not allow the part of the rotation that passes through degree-3 vertices to change between upper--middle and middle--lower alternations. Therefore, the rotations can be partitioned into three classes:
\begin{itemize}
\item the rotations with no middle edges, which form isolated vertices in the comparability graph of rotations,
\item the rotations in which every middle edge lies between two top edges, and
\item the rotations in which every middle edge lies between two bottom edges.
\end{itemize}
Each middle edge belongs to one rotation in the second class and one in the third. The rotations in the first and second classes are maximal in the partial order of rotations, and the rotations in the first and third classes are minimal. The lattice of stable matchings can be described as the lattice of lower sets of rotations, and these are just the closures in a directed graph on rotations having an edge from each rotation in the second class to each rotation sharing a middle edge with it in the third class. This graph is a bipartite orientation of the comparability graph of rotations.

In the other direction, suppose we are given a lattice of closures of a bipartite orientation of a graph $G$ and wish to find a corresponding subcubic degree-matchable graph $G'$ of stably matchable pairs for an instance of stable matching with the same lattice. Because a bipartite orientation is always a transitively reduced directed acyclic graph, the join-irreducibles of its lattice of closures are in one-to-one correspondence with its vertices, and its underlying undirected bipartite graph is isomorphic under this correspondence to the comparability graph of its join-irreducibles. When we apply the construction of \autoref{thm:dl2sub3} to this lattice, we obtain a rotation for each vertex of $G$,
with its edges labeled as top, middle, or bottom. The endpoints of middle edges of these rotations are glued to other rotations, producing degree-3 vertices in the resulting graph $G'$ of stably matchable pairs, and these degree-3 vertices have the middle edges as a perfect matching. Because no vertex of $G$ has both incoming and outgoing neighbors in the bipartite orientation, each path of top or bottom edges in each rotation has an even number of degree-2 vertices, and these vertices can be perfectly matched using  the edges of the rotation. Therefore, $G'$ is degree-matchable.
\end{proof}

\subsection{Outerplanar and series-parallel graphs}

Next, we turn to the lattices associated with outerplanar (or more generally series-parallel) graphs of stably matchable pairs. Characterizing these lattices is simplified by an additional assumption, that these graphs are subcubic, because of the following observation, which shows that in the subcubic case the definition of an intersection graph of rotations does not depend on whether we consider a rotation to be a set of vertices or of edges.

\begin{observation}
In a subcubic graph, two rotations share a vertex if and only if they share an edge incident to that vertex. Each vertex of a subcubic graph can be shared by at most two rotations.
\end{observation}

\begin{lemma}
\label{lem:prism-k33}
Suppose that for a given rotation system of a subcubic graph $G$, the intersection graph of rotations contains a cycle.
Then $G$ has at least one of $K_{3,3}$ or the graph of a triangular prism as a minor.
\end{lemma}

\begin{proof}
Let $C$ be a shortest cycle in the intersection graph of rotations and let $R$ be a rotation that is maximal among the rotations of $C$, according to the partial order of rotations. Let $S$ and $T$ be the two neighbors of $R$ in $C$. Then all of the edges shared by $R$ with $S$ and $T$ are lower edges in the alternation of upper and lower edges of $R$; in particular, no two of these shared edges can be adjacent. Because $C$ was taken to be shortest, by the observation above, the edges that $R$ shares with $S$ and $T$ are the only parts of $R$ shared with other rotations in $C$.

Let $H$ be the subgraph of $G$ induced by vertices in $C$, let $v$ be any vertex of an edge shared by two rotations of $C$, neither of which is $R$, and let $K$ be the connected component containing $v$ in $H\setminus C$. Then $v$ must have a path through $K$ to rotation $S$, and within $S$ to two endpoints of shared edges in $R$. Symmetrically, $v$ must have a path through $K$ to rotation $T$, and within $T$ to another two endpoints of shared edges in $R$. Therefore, $K$ is connected by $R$ via at least four edges to at least four distinct vertices of $R$. By combining a spanning tree within $K$ with these four connecting edges, we can find a tree in $G$ that is edge-disjoint from $R$ and has leaves at four vertices of $R$. Because $G$ is subcubic, this tree has no degree-four vertices, so we can contract it to a tree with exactly two degree-three internal nodes, still having four leaves at distinct vertices of $R$. Contracting the paths in $R$ between these four vertices produces a minor of $G$ consisting of a 4-cycle from $G$ and a tree with two internal nodes of degree three having these four vertices as its leaves. Depending on the order in which the tree leaves are attached to $R$, this minor can only be either $K_{3,3}$ or the graph of the prism.
\end{proof}

\begin{theorem}
\label{thm:outerplanar-lattice}
If a stable matching instance has a subcubic outerplanar or series-parallel graph as its graph of stably matchable pairs,
then its lattice of stable matchings is isomorphic to the lattice of closures of an oriented forest.
Every lattice of closures of an oriented forest can be realized as the lattice of stable matchings
of an instance whose graph of stably matchable pairs is subcubic and outerplanar.
\end{theorem}

\begin{proof}
An outerplanar or series-parallel graph cannot have $K_{3,3}$ or the graph of a prism as a minor. Therefore,
if a stable matching instance has a subcubic outerplanar or series-parallel graph as its graph of stably matchable pairs, by \autoref{lem:prism-k33} the intersection graph of its rotations is a forest. Each edge of this forest corresponds to a pair of rotations that share an edge in the graph of stably matchable pairs, and can be oriented from the upper of the two rotations to the lower one. With this orientation, the lower sets of the partial order of rotations are the same as the closures of the orientation of this forest.

In the other direction, every oriented forest is a transitively-reduced acyclic graph, whose  lattice of closures is the same as the lattice of lower sets of a partial order on its vertices in which $u\le v$ when there is a directed path from $v$ to $u$. Applying the construction of \autoref{thm:dl2sub3} to this lattice produces a graph in which each vertex of the forest is represented as a rotation, and in which the shared edges between these rotations correspond to edges of the forest. Gluing cycles together on shared edges in the pattern of a forest always produces an outerplanar graph, with glued edges interior and unglued edges on the outer face of its outerplanar embedding, as can be proved by induction on the number of gluing steps.
\end{proof}

\begin{figure}[t]
\centering\includegraphics[scale=0.5]{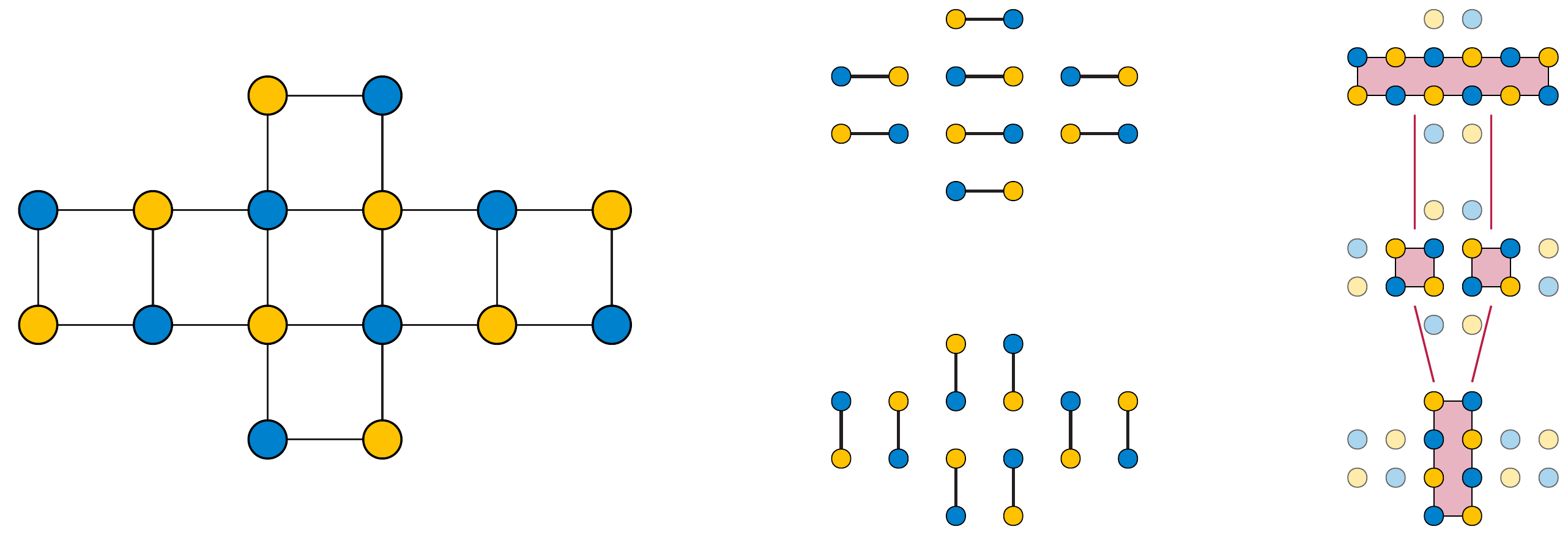}
\caption{An outerplanar graph and rotation system for which the transitive reduction of the partial order on rotations is not an oriented forest.}
\label{fig:cyclic-outerplanar}
\end{figure}

\autoref{fig:cyclic-outerplanar} shows that the assumption that the graph is subcubic is necessary for \autoref{thm:outerplanar-lattice}. Outerplanar graphs that are not subcubic may come from rotation systems whose transitive reduction contains a cycle, which would not be possible for orientations of forests.

\subsection{Planar graphs}

Finally, we turn to the lattices associated with planar graphs.

\begin{theorem}
\label{thm:planar-lattice}
If a stable matching instance has a planar graph as its graph of stably matchable pairs,
then its lattice of stable matchings is isomorphic to the lattice of closures of an oriented string graph.
Every lattice of closures of an oriented string graph can be realized as the lattice of stable matchings
of an instance whose graph of stably matchable pairs is planar and subcubic.
\end{theorem}

\begin{proof}
If a stable matching instance has a planar graph as its graph of stably matchable pairs,
find a planar embedding of it, and represent each rotation of the instance as the curve through its vertices and edges in the embedding. Every two rotations whose curves intersect are comparable in the partial order of rotations; if they intersect at an edge this is immediate while if they intersect at a vertex they are part of the totally ordered subset of rotations that pass through that vertex. Therefore, we may orient the string graph of the rotations by directing each edge from rotation $R$ to rotation $S$ whenever $S\le R$ in the partial order of rotations; with this orientation, the lower sets of the partial order of rotations coincide with the closures of the oriented string graph.

\begin{figure}[t]
\includegraphics[width=\textwidth]{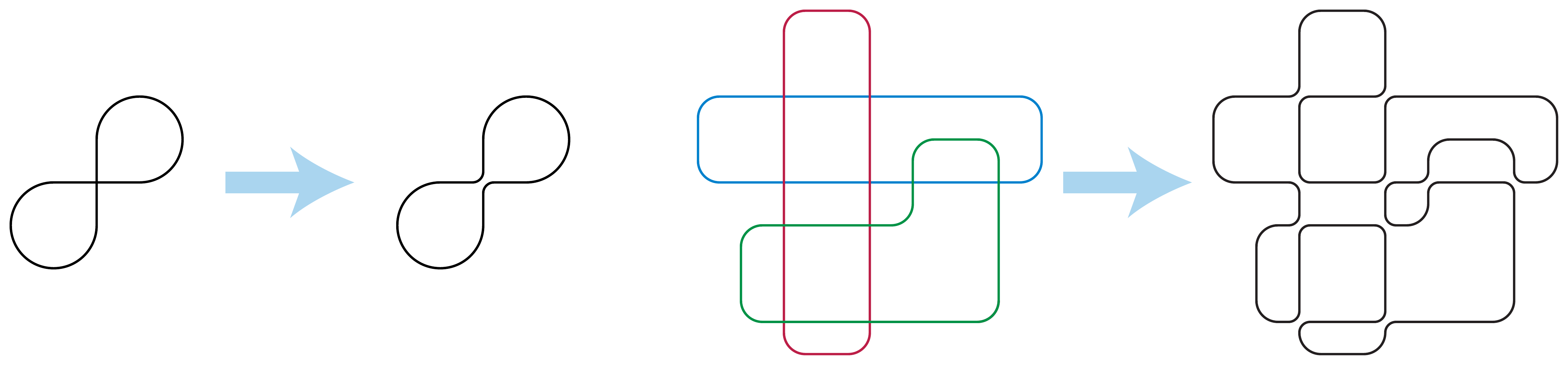}
\caption{Replacing crossings by non-crossing arcs to transform a self-crossing curve into a simple closed curve (left) and to transform a strongly connected component of an oriented string graph (orientation not shown) into a single curve (right).}
\label{fig:simplify}
\end{figure}

In the other direction, let $G$ be an oriented string graph. By doubling and perturbing each curve if necessary we may assume without loss of generality that each of the curves defining $G$ is closed (although possibly self-crossing), that each point of the plane belongs to at most two arcs of curves (either of two different curves or two arcs of the same curve), and that each point belonging to two arcs is a simple crossing point of these arcs. Whenever a curve crosses itself we may remove short lengths of the two crossing arcs, within a small disk surrounding the crossing and disjoint from the other curves, and replace the removed arcs by two non-crossing arcs while preserving the connectivity of the curve; one of the two ways of choosing how to connect these arcs will always produce a curve that remains connected.
Repeating this transformation produces a collection of curves with the same intersection graph but with no self-crossings (\autoref{fig:simplify}, left). Similarly, whenever two curves that belong to the same strongly connected component of $G$ cross, we may replace their crossing by two non-crossing arcs in a way that merges the two curves into a single curve; one of the two ways of choosing how to connect these arcs will merge them. Repeating this transformation produces a collection of curves with the same partial order on the strongly connected components of the oriented string graph but with at most one curve per component (\autoref{fig:simplify}, right). After these transformations we may assume without loss of generality that $G$ is acyclic (although not necessarily transitively reduced) and that each vertex of $G$ is represented by a simple closed curve in the plane. Because $G$ is acyclic, the join-irreducibles of its lattice of closures (which should correspond to the rotations of a rotation system representing this lattice) also correspond one-for-one with the curves in this arrangement. Whenever two curves cross, they correspond to two distinct comparable rotations.

\begin{figure}[t]
\includegraphics[width=\textwidth]{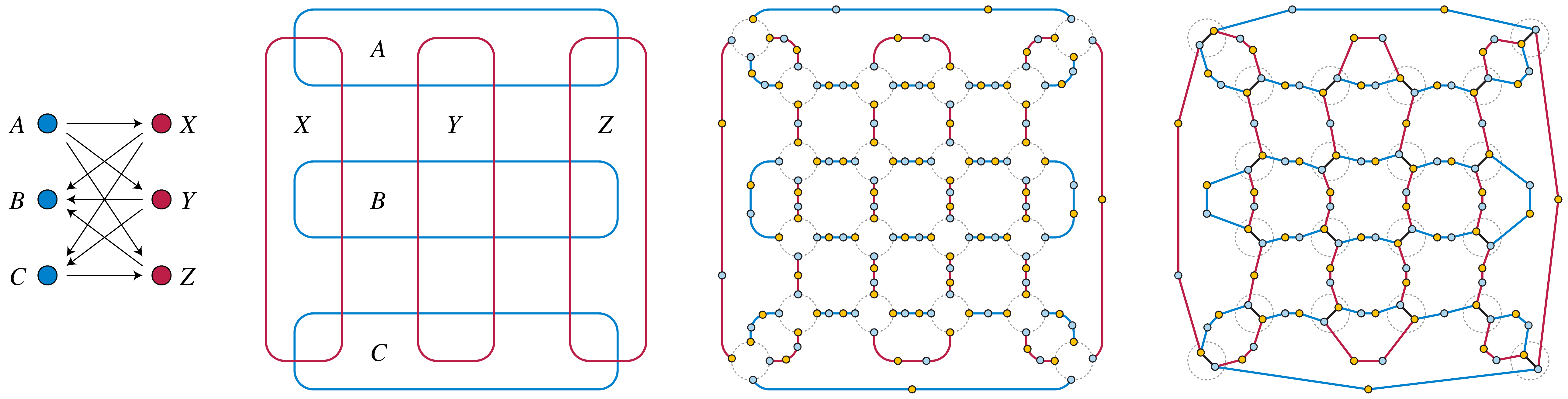}
\caption{Construction for the proof of \autoref{thm:planar-lattice}. Left: An orientation of the string graph $K_{3,3}$, and its representation as a string graph with each vertex represented as a simple closed curve. Right: adding paths of length 2 or 3 on the arcs between crossings (depending on the relative orientation of the two string graph edges at the crossings), and merging the four path endpoints near each crossing into two vertices of the final graph.}
\label{fig:string2planar}
\end{figure}

If any curve of this system of curves has no crossings, replace it by a cycle of four vertices, drawn with its edges along the curve. For the remaining curves, again remove a small disk surrounding each crossing point, partitioning each curve into a sequence of arcs connecting these disks. Along each arc of a curve $C$, draw a path of either three edges and four vertices (if the rotation corresponding to $C$ is above both of the rotations that it crosses at the ends of the arc, or below both of the rotations that it crosses) or two edges and three vertices (if the rotation is above one of the rotations that it crosses, and below the other one of these rotations, as illustrated in \autoref{fig:string2planar} (center right). By this construction, if we were to add an edge in $C$ across each of the small disks, the result would be a cycle of even length; choose arbitrarily for each curve a bipartition of the vertices of this cycle into students and residencies.

At each of the small disks where a pair of curves cross, this construction creates two students and two residencies, in the cyclic order student--student--residency--residency around the disk. Merge the two students into a single vertex within the disk, merge the two residencies into a single vertex within the disk, and add an edge between them (\autoref{fig:string2planar}, far right). The result is a subcubic bipartite graph, embedded  in the plane, and a rotation system for that graph whose rotations are the cycles following each curve of the string graph. The partial ordering on the rotations is the same as the ordering coming from the orientation of $G$. Therefore, it has a lattice of stable matchings that realizes the lattice of closures of the given oriented string graph.
\end{proof}

An example of a graph that is not a string graph can be obtained by subdividing every edge of $K_5$ into a two-edge path. If these paths are all oriented towards their central degree-2 vertex, the resulting oriented non-string graph is transitively reduced and acyclic. Its lattice of closures cannot correspond to a planar graph of stably matchable pairs,
because it is not itself an oriented string graph and it has no supergraph with the same closures.

\section{Hardness}
\label{sec:hardness}

In this section we prove that testing whether a given planar subcubic graph is a graph of stably matchable pairs is NP-complete. Our reduction is from monotone planar 2-in-4-SAT, an NP-complete problem in which Boolean variables and clauses form a planar bipartite graph, each clause is adjacent to exactly four variables, and the task is to determine whether there exists a truth assignment to the variables for which each clause is adjacent to two true and two false variables~\cite{KarKraWoo-DMTCS-07}. Because the proof is complicated, we warm up with a simpler reduction, for nonplanar subcubic graphs, from not-all-equal satisfiability with three variables per clause (NAE3SAT)~\cite{Sch-STOC-78}. This problem takes as its input a collection of Boolean variables, and clauses each of which involves three terms (variables or their negations). The output is true when there is some way to assign truth values to the variables so that each clause involves at least one true term and at least one false term.

\subsection{NAE3SAT reduction}
\label{sec:gadgets}

\begin{figure}[t]
\centering\includegraphics[scale=0.6]{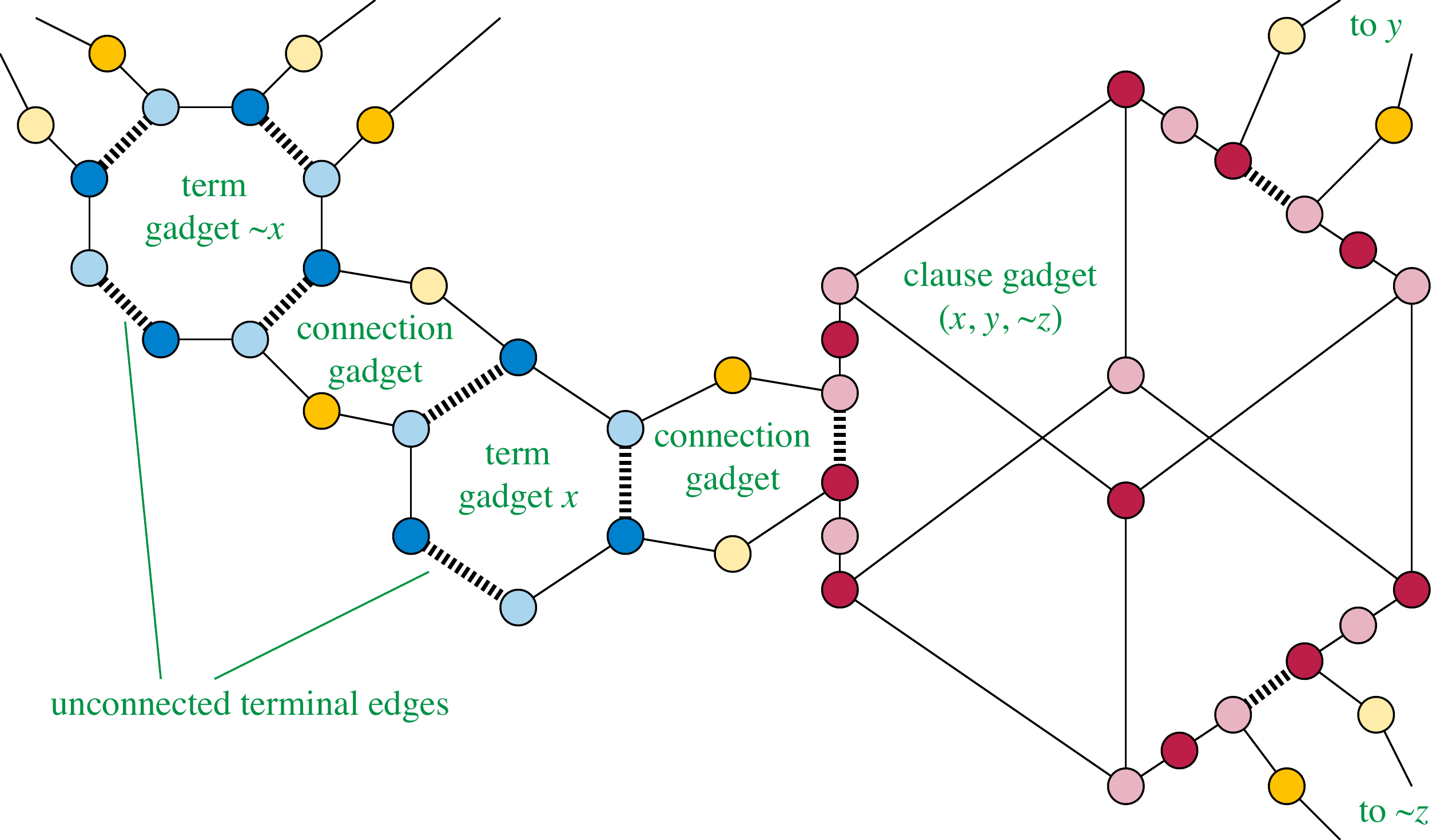}
\caption{Gadgets and their connections for the reduction from NAE3SAT to realizability as a graph of stably matchable pairs. Student vertices are shown as dark colors and residency vertices as light; term gadget vertices are blue, connection gadget vertices are yellow, and clause gadget vertices are red. The thick dashed edges are terminal edges; the truth edges are the solid edges between the blue vertices of term gadgets.}
\label{fig:redux}
\end{figure}

In this section we sketch but do not formally prove our warm-up hardness reduction from (nonplanar) NAE3SAT instances to (nonplanar) subcubic bipartite graphs.  We transform each term that appears at least once in a given NAE3SAT instance into a \emph{term gadget}, a subgraph of what will eventually become a subcubic graph. A term gadget is just a simple cycle of even length, with the vertices alternating between students and residencies, and with the edges alternating between \emph{truth edges} and \emph{terminal edges}. Our overall construction will have the property that, when the resulting graph has a rotation structure, the truth edges will belong either to the top stable matching of the structure, $\top$, or to the bottom stable matching. When the truth edges belong to $\top$, this situation corresponds to a truth assignment that makes the term true, and when the truth edges belong to $\bot$, the corresponding truth assignment makes the term false.
The terminal edges are mostly used to connect to clause gadgets for the clauses that include the given term,
but we will leave one of them unconnected, to force the truth edges to belong to either $\top$ or $\bot$, and we may also need an additional terminal edge to connect to the negation of the same clause.

For each clause in the given NAE3SAT instance, we construct a \emph{clause gadget}, in the form of a subdivision of the graph of a cube. The graph of a cube is 3-regular, with 8 vertices. It has three perfect matchings formed by three parallel edges of the cube, and six perfect matchings formed by choosing two opposite square faces of the cube and matching the vertices within each of these two faces in a way that is not parallel to the matching in the opposite face. However, there are also eight three-edge matchings in the cube that are maximal but not perfect, obtained by removing two opposite vertices of the cube and finding a perfect matching in the remaining six-vertex cycle. Our clause gadget is obtained by finding one of these three-edge maximal matchings, subdividing each of its edges into a five-edge path, and choosing the middle edge of each of these paths as one of the three terminal edges of the gadget. The cube is bipartite and subdividing edges into odd paths does not change its bipartiteness, so we can find a bipartition of it and choose arbitrarily for one side of the bipartition to be students and the other side to be residencies.

The final type of gadget that we will use is a \emph{connector gadget}, which we will use in two ways: to connect a term gadget  for a variable to the term gadget for its negation (when both the variable and its negation appear as terms in the given NAE3SAT instance), and to connect term gadgets to gadgets for the clauses in which they appear as terms.
A connector gadget consists of two paths, each of two edges, connecting the endpoints of one terminal edge to the endpoints of another terminal edge. One of the two paths connects the two student endpoints of the two terminal edges through a central residency vertex, and the other path connects the two residency endpoints through a central student vertex, ensuring that these added paths preserve bipartiteness. One of the two terminal edges should belong to a term gadget and the other should either belong to the gadget for the negation of the term or to a clause gadget.

The overall reduction, then, consists of constructing a term gadget for each term, a clause gadget for each clause, a connector gadget for each pair of terms that are negations of each other, and a connector gadget for each appearance of a term within a clause. The choice of which free edge or which terminal edge within each gadget to use for each connector edge is arbitrary, as long as each one is used at most once. These gadgets and the way they are connected are depicted in \autoref{fig:redux}.

The correctness of our reduction is based on the following analysis of our connection gadgets.

\begin{lemma}
\label{lem:connector-alternate}
Let $X$ be a connection gadget, and let $P$ and $P'$ be three-edge paths in the term or clause gadgets connected by $X$, with the center edge of each three-edge path being one of the terminal edges whose endpoints are connected by $X$. Suppose that the graph produced by the reduction of \autoref{sec:gadgets} has a rotation system, and label the edges of the graph as top, bottom, or middle, according to whether they belong to the matchings $\top$ or $\bot$ of the rotation system or to neither of these two matchings, respectively. Then one of paths $P$ or $P'$ must have edges whose labels alternate between top and middle, and the other must have edges whose labels alternate between middle and bottom.
\end{lemma}

\begin{proof}
Consider the sequence of matchings produced by starting at $\bot$ and then applying the rotations of the rotation system, one by one, in an ordering given by an arbitrary linear extension of the partial order of rotations, until reaching $\top$. Then to cover all three edges at each of the degree-three vertices of $X$, the part of the matching within subgraph $X\cup P\cup P'$ must change at least twice, passing through at least three distinct matchings, throughout this sequence. The first of these changes cannot pass through a degree-two vertex of $X$, because after a change through such a vertex its matched edge can never change again, and this would prevent the degree-three vertex at the other end of that matched edge from changing to its third match. The only other possibility is that this first change is generated by a rotation that passes through $P$, $P'$, or both, causing this path to alternate between bottom and middle edges.

Symmetrically, the last change to the subgraph $X\cup P\cup P'$, in this sequence of matchings, cannot pass through a degree-two vertex of $X$, because the edge that this degree-two vertex was matched to before the change could not have changed in any previous step, preventing its degree-three vertex at its other end from having changed enough times to cover its three incident edges. Therefore, this last change also passes through $P$ or $P'$, alternating between top and middle edges. Since at least one of $P$ or $P'$ must alternate between bottom and middle edges, and at least one of $P$ or $P'$ must alternate between top and middle edges, exactly one of the two paths must alternate in each of these two ways.
\end{proof}

\begin{figure}[t]
\centering\includegraphics[scale=0.6]{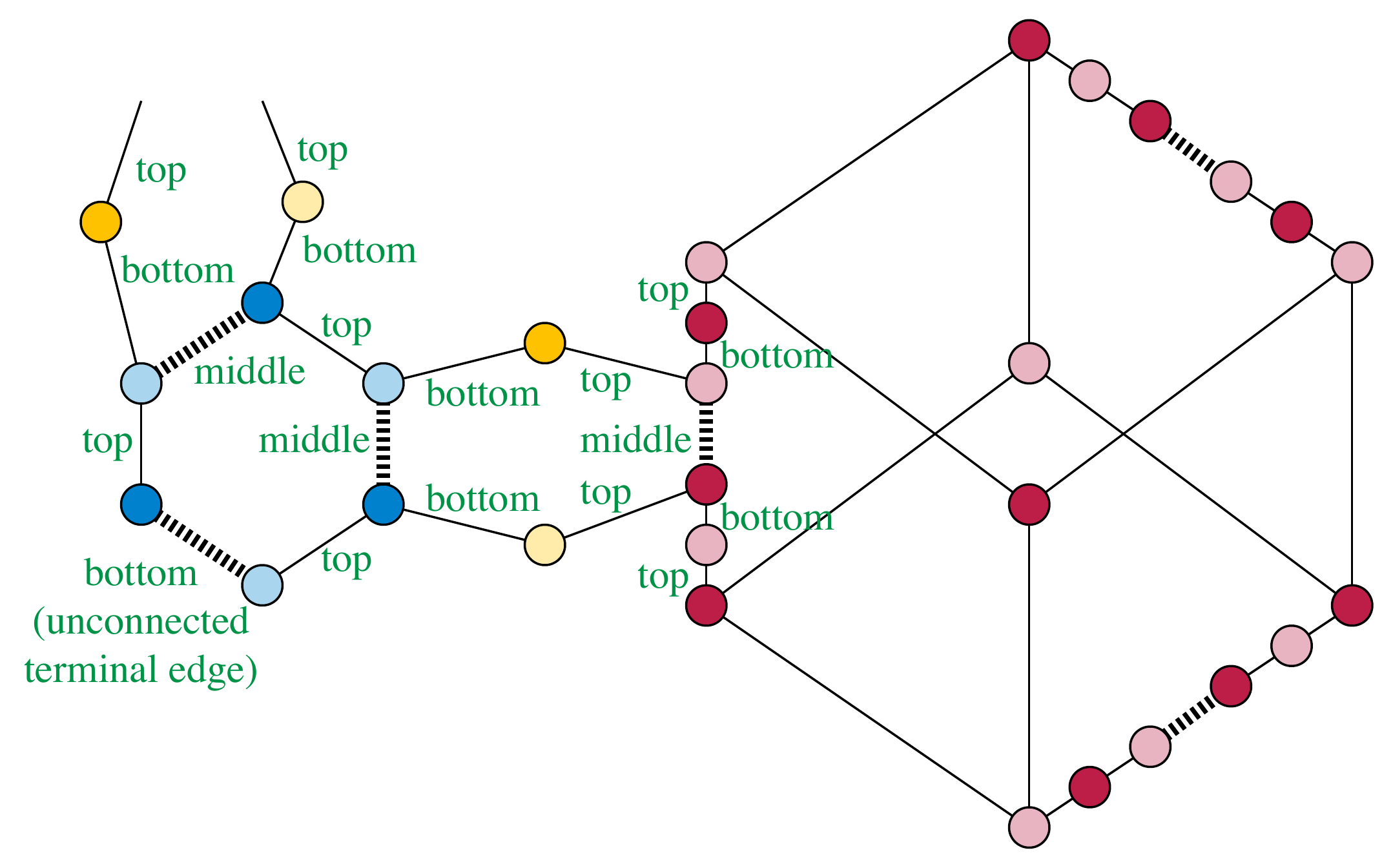}
\caption{Propagating labels of edges from term gadgets to  the adjacent connection gadgets and five-edge paths of clause gadgets}
\label{fig:propagate}
\end{figure}

If $G$ is a graph constructed as in \autoref{sec:gadgets} from an NAE3SAT instance, and the NAE3SAT instance has a satisfying truth assignment, then we can construct a valid rotation system for $G$ by labeling the truth edges of variable gadgets as top or bottom edges according to the truth assignment, assigning all terminal edges to be middle edges, and propagating these labels to the remaining graph edges consistently with \autoref{lem:connector-alternate} and with the requirement that each degree-two vertex have adjacent labels top and bottom and that each degree-three vertex have all three labels adjacent to it. \autoref{fig:propagate} shows the results of propagating these constraints on labelings from a variable gadget to nearby edges. The satisfaction of each clause in the truth assignments allows this propagation to be completed uniquely within each clause gadget. The resulting labeling corresponds uniquely to a collection of cycles in which pairs of adjacent edges alternate between the closest higher label (bottom to middle or middle to top at a degree-three vertex, bottom to top at a degree-two vertex) and the closest lower label (defined symmetrically). These cycles have a partial ordering with the central cycle of a true vertex gadget as the topmost elements, cycles in connector gadgets to true vertex gadgets next, cycles within clause gadgets in the middle,  cycles in connector gadgets to false vertex gadgets second-to-last, and the central cycles of false vertex gadgets last. Within each clause gadget the partial ordering (on the two alternating cycles within the gadget) is the same as in a cube graph with the same labels, using the outermost label of each subdivided edge as the label for the corresponding unsubdivided edge in a cube graph. Therefore, when the NAE3SAT instance is satisfiable, $G$ has a rotation system.

Conversely, when $G$ has a rotation system, \autoref{fig:propagate}  shows that the term gadgets have truth edges that must be labeled consistently according to some truth assignment. The constraints on edge labels at degree-2 and degree-3 vertices imply that the labels on the connector gadgets and the subdivided edges of clause gadgets must be as shown in \autoref{fig:propagate}. And in order for the rest of the clause gadget to have a consistent labeling, each clause must be satisfied, for an unsatisfied clause would have a clause gadget with all three subdivided edges labeled the same way, preventing $\bot$ or $\top$ from being completed to a perfect matching within that gadget. Therefore, when $G$ has a rotation system, the NAE3SAT instance is satisfiable.

Although this is enough to prove NP-completeness of recognizing subcubic graphs of stably matchable pairs,
it does not prove the stronger result that we have claimed, NP-completeness for planar subcubic graphs.
There are several difficulties that we must overcome in order to prove this stronger result.
\begin{itemize}
\item Although our clause gadgets (subdivided cubes) are planar graphs, the terminal edges by which they connect to the rest of the reduction do not all lie on a single face of their planar embedding, causing the overall reduced graph to become nonplanar even when starting with a planar instance of NAE3SAT. Therefore, we need alternative clause gadgets that accomplish a similar function while preserving planarity.
\item Planar NAE3SAT is not a hard problem; it has a polynomial-time algorithm. So we cannot merely replace the clause gadgets of our NAE3SAT reduction; we need a different hard problem. Realizability by stable matchings has an inherent symmetry (swapping $\top$ and $\bot$ and reversing the partial order on rotations preserves the validity of a rotation structure), and this symmetry is reflected in the symmetry of inverting all truth values in NAE3SAT. For our planar reduction, we have chosen 2-in-4-SAT, which is also symmetric under truth value inversion, but has more complicated clauses with four terms instead of three.
\item Our NAE3SAT reduction preserves the bipartiteness of the resulting graph by leaving a choice free for each connection gadget, of which pairs of vertices to connect by paths. One of the two choices can be used, regardless of  previously made choices of which of the connected vertices are students and which are residencies. But if the resulting graph is to be planar, we do not have this freedom, because one of these two choices will introduce an undesired crossing. Instead we must make sure that when all the gadgets we are connecting have their bipartitions chosen consistently. This is already a problem for the term gadgets in our NAE3SAT reduction, because when we use a connector gadget to connect a term and its negation, and embed the resulting planar graph, the two terms will be oppositely oriented (one with students clockwise from residencies on its terminal edges, the other with students counterclockwise from residencies). So we need a gadget that, like the connector gadget, can reverse the truth value associated with a connection, but without inverting its orientation.
\item In our NAE3SAT reduction, partially ordering the rotations was aided by the fact that the term gadgets are inherently acyclic (either maximal or minimal in the partial order of rotations, rather than having both predecessors and successors in the partial order). This might not be true for other gadgets; for instance, it is not true of the clause gadgets. So if other gadgets are to be connected together we need to take care that yes-instances have rotations that can be partially ordered. On the other hand, as we will see, the possibility that a collection of cycles might fail to be a rotation system by having a cycle in its ordering can lead to new opportunities for constraining the behavior of gadgets.
\end{itemize}

\subsection{More gadgets}

\begin{figure}[t]
\centering\includegraphics[scale=0.6]{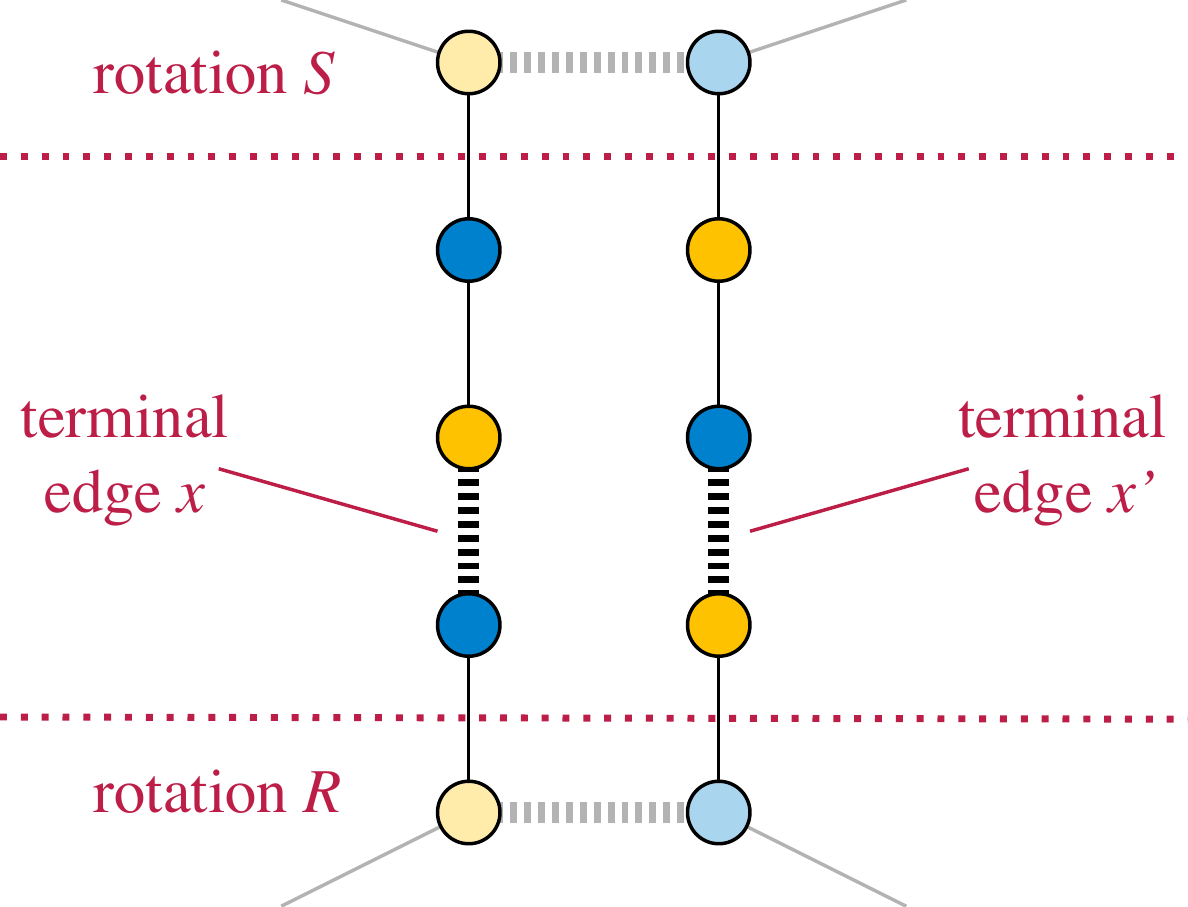}
\caption{The switch gadget}
\label{fig:switch}
\end{figure}

We will re-use the term gadgets and connection gadgets from the NAE3SAT reduction, as analyzed using \autoref{lem:connector-alternate}. As in that reduction, gadgets will be connected to each other at \emph{terminal edges}; to preserve planarity, we will usually require these terminal edges to be on the outer face of an embedding of each gadget. To handle the issue that the two paths of a connection gadget cannot cross, we will define the \emph{polarity} of a terminal edge to be either clockwise or counterclockwise, the direction around the outer face of the gadget followed by orienting the terminal edge from its student endpoint to its residency endpoint. The polarity of all terminal edges of a single gadget can be reversed by taking the mirror image of the gadget, but it cannot be changed for individual terminal edges without changing the polarity of all other terminal edges of the same gadget. All terminal edges of a term gadget have the same polarity as each other, but the two edges of a connector gadget have opposite polarity.

We add to these the following repertoire of gadgets and combinations of gadgets.

\begin{description}
\item[The switch gadget.] This is, essentially, a connection gadget with its paths extended by two more edges, allowing the interior edges of the paths to function as terminal edges to which another connection gadget can be attached, at the points marked as $x$ and $x'$ in \autoref{fig:switch}. In turn, the switch gadget connects to the terminal edges from two other gadgets. Following the same analysis as \autoref{lem:connector-alternate}, these two other gadgets must have rotations ($R$ and $S$ in the figure) that pass through their terminal edges, one of which is earlier than the central cycle of the switch gadget in the partial ordering of rotations, and the other of which is later. The key behavior of the switch gadget is to control the relative orderings of these  two rotations: if the incoming connection gadget at $x$ has top edges at the ends of its paths, then the path within the switch gadget must have edges labeled (in order from $R$ to $S$) bottom-middle-bottom-top, forcing $R$ to be earlier than $S$ in the partial order of rotations. If the connection gadget attached to $x$ has bottom edges at the ends of its paths, then symmetrically $S$ must be earlier than $R$ in the partial order of rotations.

To achieve this behavior, only one of the two terminal edges $x$ and $x'$ needs to be attached to a connection gadget. If both are attached to connection gadgets, then both connection gadgets are forced to have the same state, which can be used to propagate information from one side of the switch gadget to the other. In some ways this is like a crossover gadget in some other planar NP-completeness proofs, but unlike a crossover gadget there are not two independent bits of information crossing each other: the partial ordering of $R$ and $S$ carries the same information as the state of the connection gadgets at $x$ and $x'$.

If terminal edge $x'$ is unused, the path containing it can be shortened to two edges without affecting the behavior of the rest of the gadget.

\begin{figure}[t]
\centering\includegraphics[scale=0.6]{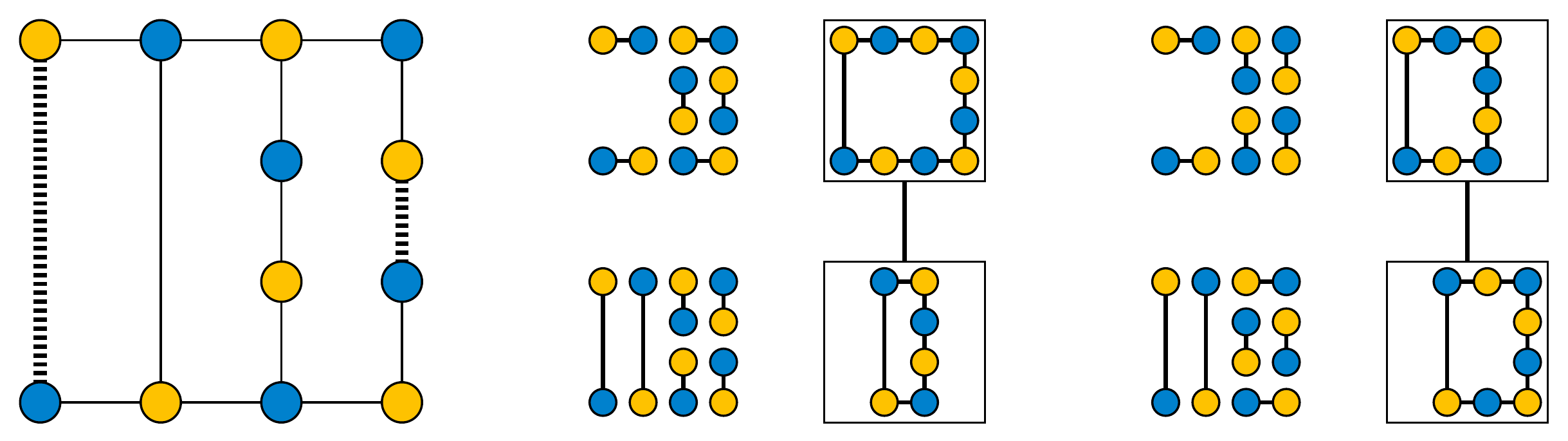}
\caption{The fuse gadget and two of its four rotation systems. The other two are obtained from the two shown by swapping the top and bottom matchings and reversing the order of the rotations.}
\label{fig:fuse}
\end{figure}

\item[The fuse gadget.] This is obtained from a $2\times 4$ grid by subdividing some of its edges into paths of odd lengths, whose degree-two vertices force the subdivided paths to use only top and bottom edges. The unsubdivided $2\times 4$ grid has rotation systems consisting of the three faces of the grid, in some order, but the subdivision prevents these rotation systems from being valid. Instead, the remaining four rotation systems are of two types, as shown in \autoref{fig:fuse}.

We will always attach switch gadgets to the terminal edges of the fuse gadgets. To understand the behavior of a fuse gadget, suppose that one of its two attached switch gadgets connects a rotation $R$ (at its other end) to a rotation $S$ within the fuse gadget, in an ordering controlled by a connection gadget at terminal edge $x$. Suppose symmetrically that the other attached switch gadget connects a rotation $S'$ within the fuse gadget, to a rotation $T$ at its other end, in an ordering controlled by a connection gadget at terminal edge $y$. 

Two of the rotation systems for the fuse gadget (including the one shown in the central part of the figure) use the outer boundary of the grid as one of the rotations, and the center face of the grid as the other rotation, causing one of the two terminal edges to be upper and the other to be lower in the same rotation as each other. The outer rotation will act as both $S$ and $S'$ for the two attached switch gadgets, while the inner rotation is separated from the switch gadgets.
When connection gadgets $x$ and $y$ have the same state, exactly one of these two rotation systems will be valid, producing one of two transitive orderings of rotations: $R\le S=S'\le T$, or $R\ge S\ge T$. Intuitively, when the switch gadgets are both switched in the same direction, the fuse gadget transmits the ordering relation from one switch gadget to the other.

The other two rotation systems for the fuse gadget are the one shown on the right of the figure and its reverse. When the two attached switch gadgets are switched in opposite directions, one of these two rotation systems will be valid within the fuse gadget, and there will be no order relation between $S$ and $T$ caused by these gadgets.

\begin{figure}[t]
\centering\includegraphics[scale=0.4]{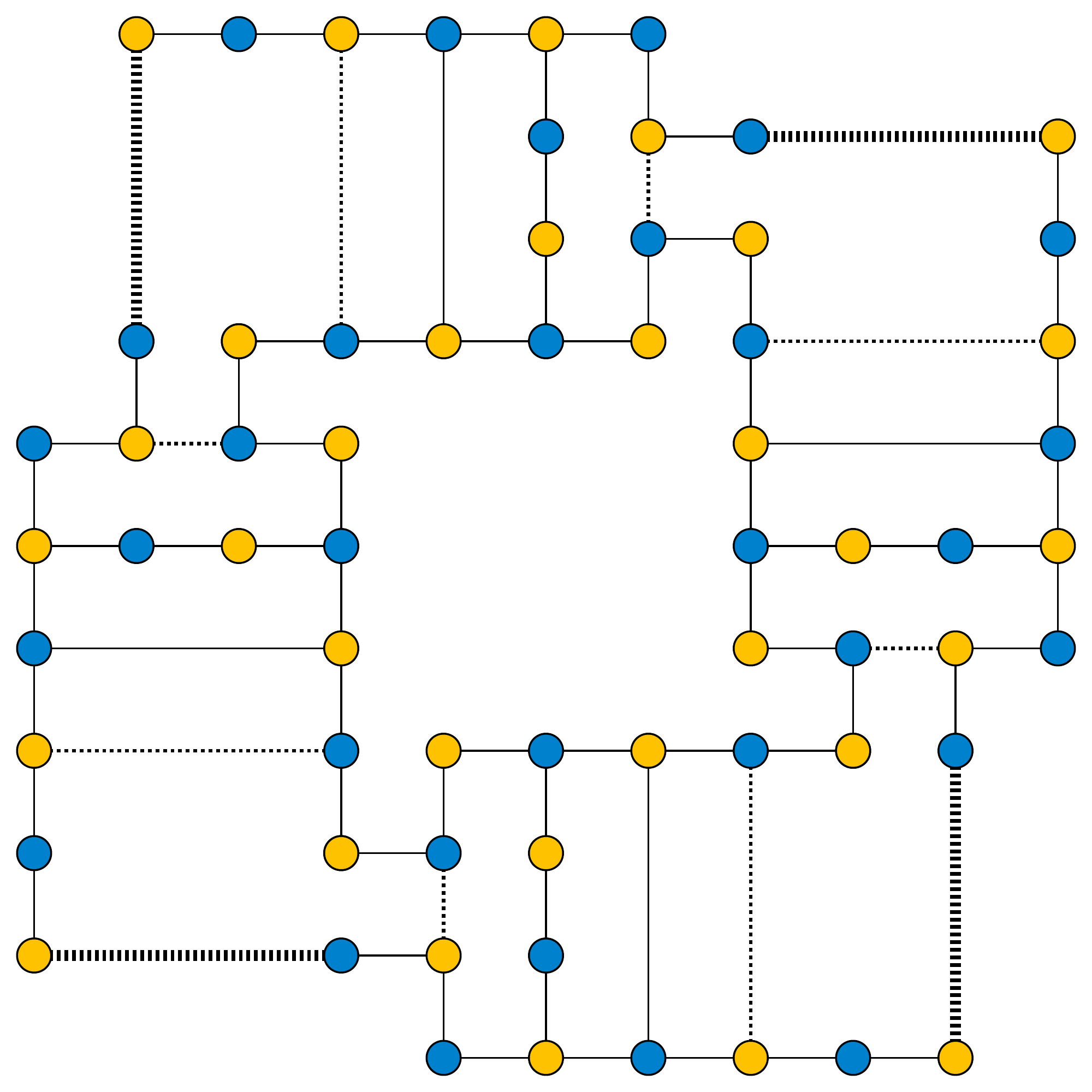}
\caption{Planarity-preserving four-input not-all-equal gadget, constructed from a cycle of four switch gadgets and four fuse gadgets.}
\label{fig:planar-nae}
\end{figure}

\item[Planar not-all-equal gadgets.] We can construct a not-all-equal gadget with any number $k$ of inputs by connecting switch gadgets and fuse gadgets in a cycle, as shown in \autoref{fig:planar-nae}. Because both switch gadgets and fuse gadgets reverse the polarity of the two terminal edges that they attach to, the result is bipartite even when $k$ is odd, and all of its free terminal edges (the terminal edges used to connect connector gadgets to switch gadgets) have the same polarity as each other.

When the connector gadgets attached to these switch gadgets all have the same state (all edges at the attached ends of their paths labeled as top, or all labeled as bottom), the combined behavior of the switch and fuse gadgets will cause the not-all-equal gadget to have a cycle in its ordering of rotations, preventing a valid rotation structure from existing. When not all connector gadgets have the same state, this cycle will be broken at the fuse gadget between each two consecutive connectors that have unequal states, allowing the existence of a valid and partially ordered set of rotations within this part of the overall graph.

\item[Inverter gadget.] In our NAE3SAT reduction, we used connector gadgets between pairs of a term gadget and its negation, which in this application acted as an inverter: if the edges at one end of a connector gadget's paths are labeled as top, the edges at the other end are labeled as bottom. However, a connector gadget also reverses the polarity of its two terminal edges, which may be problematic in some of our constructions. We may obtain an inverter that does not reverse polarity as a special case of our not-all-equal gadget with two inputs (a cycle of two switch gadgets and two fuse gadgets).

In any valid rotation system involving the inverter gadget, the cycles within the two switch gadgets will have no ordering relation to each other, because they are isolated from each other by the fuse gadgets between them and by the requirement that the two switch gadgets be ordered in opposite directions at each fuse gadget. Therefore, it is possible to use this gadget without any concern that there may be additional ordering relations on rotations transmitted through the connection gadgets attached to it.

\item[Insulator gadget.] As discussed above, the inverter gadget attaches to two connector gadgets that must have opposite states, and breaks any ordering relation that might have been created between rotations at the other ends of these two connectors. It is useful to have a gadget that, similarly, attaches to two connector gadgets that must have equal states, and breaks any ordering relation that might have been created between rotations at the other ends of these two connectors. This can be accomplished by using a term gadget with two terminal edges (not counting its unconnected terminal edge).
\end{description}

\subsection{Nested gadget cycles for arbitrary inversion-symmetric clauses}

By combining the simpler gadgets above, we have a general construction for gadgets that can represent any logical function that is symmetric under the inversion of the truth values of all its variables (as the not-all-equal function is).
As in our NAE3SAT construction, define a truth edge of a gadget to be an edge adjacent to a terminal edge, which (in all rotation systems for all our gadgets) must be either part of the top or bottom matching, and define such an edge to correspond to a Boolean value, true if it belongs to the top matching and false if it belongs to the bottom matching.

\begin{lemma}
\label{lem:symmetric-clause}
Let $f$ be a nontrivial Boolean function of any number $k$ of Boolean variables, with the property that inverting all variables does not change the value of $f$. Then there exists an embedded planar subcubic bipartite graph $G$, with $k$ designated terminal edges of equal polarities and degree-2 endpoints on the outer face of its embedding, corresponding to the variables of $f$, with the following property. For every truth assignment to the variables of $f$, and every corresponding labeling of the truth edges of $G$ as belonging to the top or bottom matchings of a rotation system, $f$ is true for the given truth assignment if and only if $G$ has a rotation system extending the given labeling. Further, whenever a truth assignment corresponds to a rotation system, the rotation system can be chosen in such a way that there are no order relations in its partial order between any two rotations containing terminal edges.
\end{lemma}

\begin{proof}
For $k$ Boolean variables $x_1,\dots x_k$, there are $2^{k-1}$ complementary pairs of truth assignments, some of which cause $f$ to become true and some of which cause it to become false. We prove the result by induction on the number of pairs of truth assignments that cause $f$ to become false. As a base case, when there is only one forbidden truth assignment, we may use the not-all-equal gadget of \autoref{fig:planar-nae}, with each of its terminal edges attached by a connector gadget to either an inverter gadget or an insulator gadget. The choice of whether to attach an inverter or insulator can be made by choosing an insulator for variables that have the same value as $x_1$ in the forbidden truth assignments, and an inverter for variables that have the opposite value to $x_1$.
The not-all-equal gadget itself is unrealizable only for the all-equal truth assignment, and this choice of inverter gadgets instead causes the graph to become unrealizable only for the single forbidden truth assignment of $x$. The use of inverter and insulator gadgets at every terminal edge of the central not-all-equal gadget prevents rotation systems for this graph from having order relations between the rotations containing its terminal edges.

\begin{figure}[t]
\centering\includegraphics[scale=0.5]{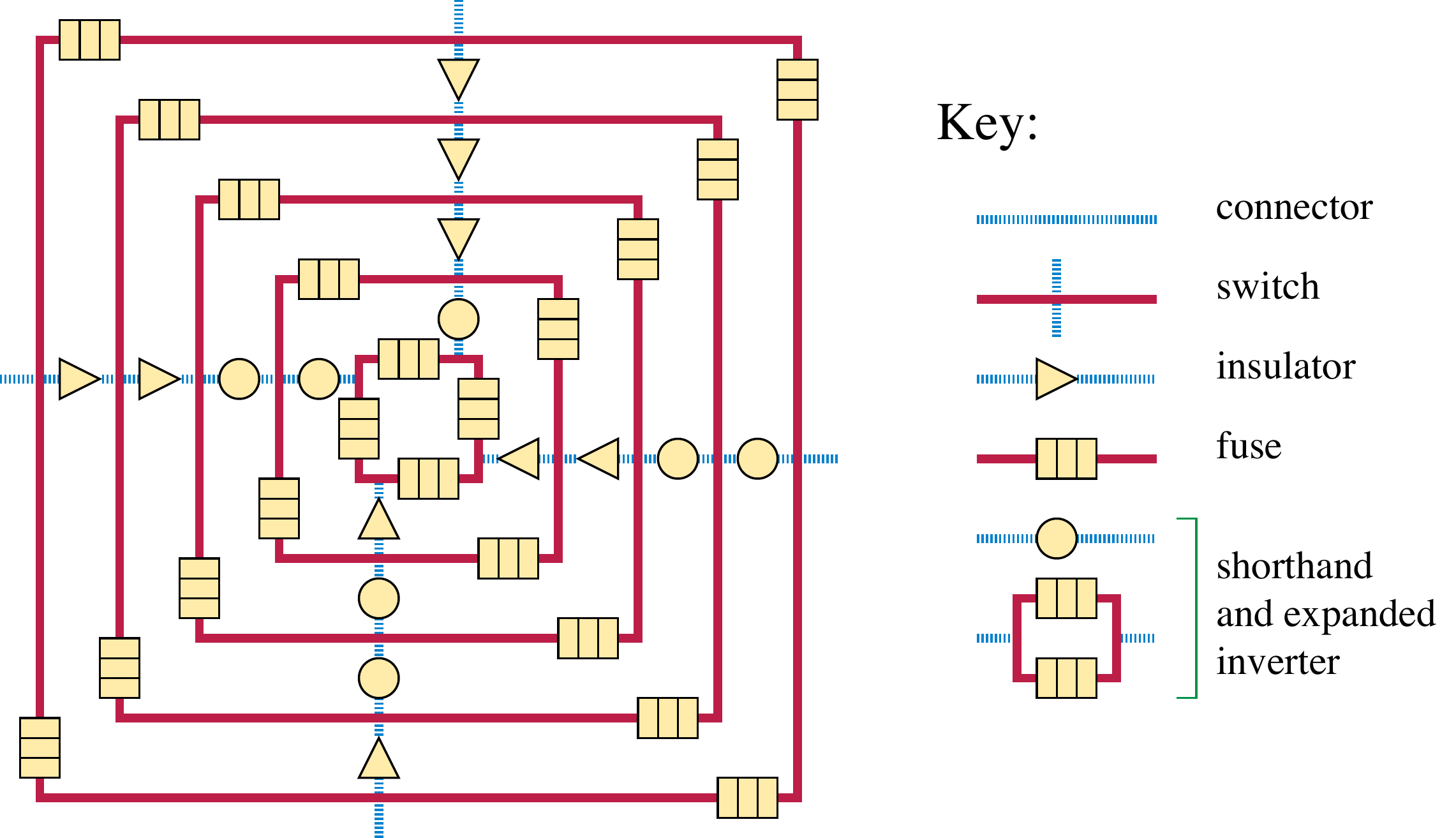}
\caption{Construction of a 2-in-4 clause gadget from simpler gadgets, according to \autoref{lem:symmetric-clause}. The innermost cycle of four switch gadgets and four fuse gadgets in this figure, when expanded from gadgets into vertices and edges, produces a subgraph isomorphic to the entirety of \autoref{fig:planar-nae}.}
\label{fig:2in4}
\end{figure}

When $f$ has more than one forbidden truth assignment, choose one of them, let $f'$ be the function forbidding only that truth assignment, construct a graph $G'$ for it as above, and let $f''$ be the Boolean function formed from $f$ by changing its value to true on the forbidden assignment for $f$ and then complementing its variables in the same pattern that was used for the inverters in the construction of $G'$. By induction, construct a graph $G''$ corresponding to the function $f''$, and embed $G''$ inside the central face of $G'$, attached by connector gadgets to the inner terminal edges of the switch gadgets of $G'$. \autoref{fig:2in4} shows this construction for the 2-in-4 function $f$, which forbids five of the eight complementary pairs of truth assignments to its four Boolean variables (the pair of all-equal assignments, and four pairs of assignments that set one of the four variables opposite to the other three). The figure omits the outermost layer of insulator gadgets, as they are unnecessary when this 2-in-4 gadget is incorporated into our overall reduction.

In the resulting graph $G$, each connection gadget, switch gadget, inverter gadget, and insulator gadget has two terminal edges, whose adjacent truth edges can be labeled in only one way consistent with the labeling at the outer terminal edges of $G$: at each connection gadget, switch gadget, or insulator gadget, the two terminal edges must have the same labeling at their incident truth edges, while at an inverter gadget the labels must be opposite. This labeling determines the structure of the rotations within all of these gadgets, leaving only the fuse gadgets undetermined. For a truth assignment that satisfies $f$, all of the cycles of switch and fuse gadgets in $G$ will have at least two consecutive switch gadgets with opposite states, breaking the potential cyclic ordering of rotations at the fuse gadget between these two switches and producing a valid rotation system for $G$ meeting all the conditions of the lemma. For a truth assignment that does not satisfy $f$, one of the cycles of switch and fuse gadgets, the one corresponding to the complementary pair including that assignment, will have all four of its switch gadgets in the same state. In this case, the only locally-valid edge labeling to the switch and fuse gadgets of this cycle produces a cyclically-ordered set of rotations, so there is no valid global rotation structure.
\end{proof}

\subsection{The overall reduction}

\begin{theorem}
\label{thm:npc}
It is NP-complete, given a bipartite graph $G$, to determine whether there exists a stable matching instance for which $G$ is the graph of stably matchable pairs, even when $G$ is planar and subcubic.
\end{theorem}

\begin{proof}
To show that recognizing graphs of stably matchable pairs belongs to NP, we may either use a preference system for the stable matching instance as a witness, with a verification algorithm that constructs the rotation system from the preference system and checks that its graph equals $G$, or we may use a rotation system directly as a witness, with a verification algorithm that checks that it meets the definition of a rotation system and has $G$ as its graph.

To show that the problem is NP-hard, we find a polynomial-time many-one reduction from the known NP-hard monotone planar 2-in-4-SAT problem to the problem of the theorem. This reduction finds a planar embedding of the given monotone planar 2-in-4-SAT instance, replaces each variable of the instance by a term gadget, replaces each 2-in-4 clause of the instance by a 2-in-4 clause gadget according to \autoref{lem:symmetric-clause} and \autoref{fig:2in4}, and replaces each connection between a variable and a clause in the planar embedding of the instance by a connection gadget, attached to the terminal edges of the term and clause gadgets in the cyclic ordering given by the planar embedding.

When the 2-in-4-SAT instance has a satisfying truth assignment, this assignment may be reflected to a labeling of the truth edges of the term gadgets as being top or bottom edges, and this labeling propagated to the remaining gadgets as described in \autoref{lem:symmetric-clause}, producing a collection of rotations that are locally partially ordered and that have no nonlocal transitive relations in their ordering. The result is a valid rotation system for the reduced graph.
Conversely, when the reduced graph has a rotation system, the truth edges of the term gadgets must be consistently labeled according to a valid truth assignment, and the unique consistent propagation of this labeling to the other gadgets, as described in \autoref{lem:symmetric-clause}, can only produce a value rotation system when this truth assignment satisfies all the clauses, so the starting 2-in-4-SAT instance is satisfiable.

We have described a polynomial-time many-one reduction from monotone planar 2-in-4-SAT to recognizing planar subcubic graphs of stably matchable pairs, and we have shown that  recognizing arbitrary graphs of stably matchable pairs belongs to NP, so both the recognition problem for arbitrary graphs and its special case for subcubic planar graphs are NP-complete.
\end{proof}

Because recognizing graphs of stably matchable pairs is trivial on cubic graphs (\autoref{obs:regular}) but hard on subcubic graphs, it provides another answer to a 2012 question of Vinicius dos Santos, who asked for natural problems that are polynomial on cubic graphs but NP-complete for subcubic graphs~\cite{VdS-TCSSX-12}.

\section{Exact algorithms}
\label{sec:algorithms}

\begin{figure}[t]
\centering\includegraphics[scale=0.35]{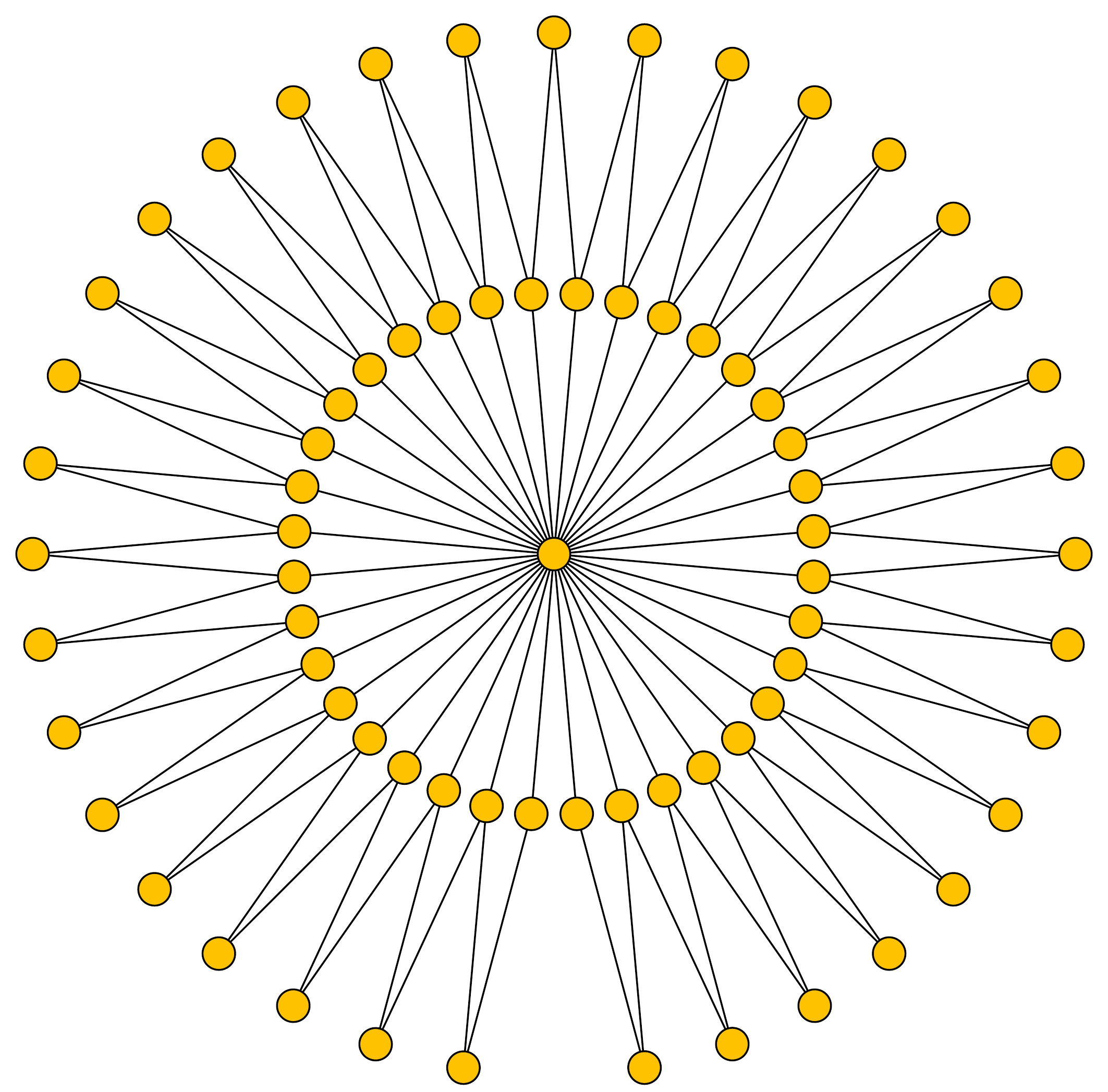}
\caption{Construction for a family of biconnected bipartite outerplanar graphs (graphs of stably matchable pairs by \autoref{thm:outerplanar}) in which a brute force search over all preference systems would take factorial time.}
\label{fig:factorial}
\end{figure}

Although recognizing the graphs of stably matchable pairs is NP-complete, it may still be possible to find algorithms for this problem that are faster than a naive brute force search through all preference systems. We provide two such algorithms. Our first algorithm, for arbitrary graphs, runs in time $2^{\Theta(m)}$ where $m$ is the number of edges, improving the $2^{\Theta(m\log m)}$ time bound that one would obtain from a brute force search (see \autoref{fig:factorial}). Our second algorithm shows that recognizing the graphs of stably matchable pairs is fixed-parameter tractable for graphs of bounded carving width, and takes linear time in such graphs (with a constant factor depending exponentially but not singly-exponentially on the carving width).

We will assume, as is standard for the design and analysis of algorithms, the unit cost model of computation in which arithmetic and Boolean operations on binary values within a single machine word, and operations that access a single machine word of computer memory, take constant time. We will also necessarily assume that each machine word stores enough bits to address the amount of memory needed for our algorithms. For algorithms whose storage requirements are exponential in the input size, this means that the entire input may be assumed to fit into a machine word or a small number of words. It is not reasonable to intrepret this as suggesting the existence of computers whose word size grows with the problem input size; instead, our assumptions should be interpreted as meaning that (as is generally true in practice) our algorithms can only be applied to inputs small enough that they do not overflow the available memory. However, the choice between these two interpretations makes no difference to the mathematics of the problem.

\subsection{Singly exponential in the number of edges}

Our singly exponential time and space bounds will depend on the following analysis of matching-related structures in graphs.

\begin{lemma}
\label{lem:matchplus}
Let $T$ be a type of structure associated with a graph, and let $c_T$ be a constant, such that a graph with $m$ edges has $O(c_T^m)$ structures of type $T$. Let $t(m)$ denote the maximum number of pairs of a perfect matching and a structure of type $T$ in the subgraph of unmatched edges, maximized over graphs with $m$ edges. Then $t(m)=O(c_M^m)$, where $c_M$ is at most $1+c_T$ and can be upper-bounded as
\[
c_M \le \max\left\{ \left(i\, c_T^{2i-2}\right)^{1/(2i-1)} \mid 1\le i \right\}.
\]
\end{lemma}

\begin{proof}
By the binomial theorem, the number of structures consisting of a set of edges, together with a structure of type $T$ within that set, is
\[
\sum_{i=0}^m \binom{m}{i} O(c_T^i) = O\bigl( (1+c_T)^m \bigr).
\]
The bound $c_M\le 1+c_T$ follows by considering the set of edges in such a structure to be the set complementary to a perfect matching.

To get the more precise bound of the lemma, we can list structures formed as a perfect matching plus a structure of type $T$ by the following algorithm: repeatedly find a vertex of minimum degree in the current graph, and loop over its incident edges. For each such edge $e$, recursively call the same algorithm in the subgraph formed by removing the two endpoints of $e$ and all their incident edges. If the recursion ever finds a vertex of degree zero, it returns without listing anything (as the choices made in this branch of the recursion do not lead to a perfect matching). Otherwise, the bottom-level calls to this algorithm (in which the remaining subgraph is empty) have each found a perfect matching, and list all structures of type $T$ (in an unspecified way) in the subgraph of removed but unmatched edges.

At a recursive call of this algorithm where the minimum degree is $i$, the number of unmatched removed edges that will be considered later in listing structures of type $T$ is at least $2i-2$, and the total number of removed edges is at least $2i-1$, explaining the two exponents in the formula of the lemma. (The cases where more than this number of unmatched edges are removed produce similar terms with greater exponents, but they can be omitted from the formula as they cannot produce the maximum value of the overall formula.) By choosing $c_M$ as upper-bounded by the formula of the lemma, it follows by induction on $m$ that the total number of structures found by this algorithm is as stated.
\end{proof}

We observe that in cases where the structures of type $T$ can be efficiently enumerated (for instance when they are just subsets, as below) the proof of the lemma can be made into an algorithm whose running time is obtained by replacing the bound on the number of structures of type $T$ by the running time for enumerating these structures.

\begin{lemma}
\label{lem:nmatch}
A graph with $m$ edges has at most $c_0^m$ perfect matchings, where
\[
c_0 = 6^{1/9} \approx 1.22028.
\]
This bound is tight for a graph consisting of disjoint copies of $K_{3,3}$.
\end{lemma}

\begin{proof}
This follows from a result of Alon and Friedman \cite{AloFri-EJC-08} according to which the number of perfect matchings in a graph with degree sequence $d_i$ is at most $\prod (d_i)^{1/2d_i}$. In a graph with $m$ edges, there are $2m$ vertex-edge incidences, which can be grouped in various ways to vertices of varying degrees. According to the formula of Alon and Friedman, an incidence that participates in a vertex of degree $d$ contributes a factor of $(d!)^{1/2d^2}$ to the total bound on the number of matchings. This factor is maximized at degree three, for which it is $6^{1/18}$. If each vertex-edge incidence participates in a vertex of degree three, maximizing the contribution to the product from every incidence, then the total product is as stated.
\end{proof}

It is not possible for all matchings of $K_{3,3}$ to be stable so this provides an upper bound but not a tight bound on the number of stable matchings as a function of~$m$.

\begin{lemma}
\label{lem:exponentials}
In a graph with $m$ edges, the number of structures consisting of a perfect matching and an arbitrary subset of the remaining edges is $O(c_1^m)$, and the number of structures consisting of two perfect matchings (not necessarily disjoint) and an arbitrary subset of the remaining edges is~$O(c_2^m)$, where
\[
c_1 = 1280^{1/9} \approx 2.21435
\]
and
\[
c_2 = (4c_0(1+c_0)^6)^{1/7} \approx 2.48475.
\]
\end{lemma}

\begin{proof}
The bound on $c_1$ follows from \autoref{lem:matchplus}, counting matchings plus a structure $T$ consisting of an arbitrary set of edges, with $c_T=2$. Calculation shows that the maximum value for the upper bound in the lemma is achieved at $i=5$, giving the value shown.

The bound on $c_2$ follows by using the same minimum-degree algorithm from the proof of \autoref{lem:matchplus} to choose a matching and a disjoint set of edges, and then applying \autoref{lem:nmatch} to bound the number of matchings in the set of edges that were not chosen to be included in the disjoint set of edges. Each matched edge or unmatched but unchosen edge contributes a factor of $c_0$ to the total number of structures found by this combination of bounds. When the minimum vertex has degree $i$ and all its neighbors, this produces a contribution per removed edge of $\bigl(ic_0(1+c_0)^{2i-2}\bigr)^{1/(2i-1)}$, almost the same as if we were applying \autoref{lem:matchplus} with $c_T=1+c_0$. Here, the first factor of $c_0$ comes from the application of \autoref{lem:nmatch} to the edge that was chosen to be in the matching, and the factor of $(1+c_0)^{2i-2}$ comes from applying the binomial theorem to the choice of edges to be part of the disjoint set and to the factor of $c_0$ that will be included when an edge is not chosen.  When the minimum vertex has degree $i$ but some of its neighbors have higher degrees, the leading $ic_0$ factor in this formula is averaged over a greater number of edges each contributing a factor of $1+c_0$, leading to a more complicated calculation but a lower contribution per removed edge than for the case of equal degrees. Calculation shows that the maximum value for this upper bound is achieved for a vertex with neighbors of equal degree, at $i=4$, giving the value stated in the lemma.
\end{proof}

\begin{theorem}
\label{thm:single-exponential}
We can test whether a given graph $G$ with $m$ edges and $n$ vertices is a graph of stably matchable pairs in space $O\bigl(\min\{c_1^m,(n/2)!2^{m-n/2}\}\bigl)$ and time $O\bigl(\min\{mc_2^m,(n/2)!^2 2^{m-n/2}\}\bigl)$, where $c_1$ and $c_2$ are the numbers from \autoref{lem:exponentials}.
\end{theorem}

\begin{proof}
We will search for a sequence of perfect matchings, differing from one to the next in the sequence by a single rotation, that together cover all edges of $G$ and such that each edge of $G$ is rotated into a matching at most once and rotated out of a matching at most once. If such a sequence exists, we can produce a rotation structure that has the first matching in the sequence as its bottom matching, the last matching in the sequence as its top matching, and the rotations of the sequence as the rotations of the rotation structure, partially ordered consistently with the sequence order.

We search for this sequence as a path in a graph $\Gamma$ whose vertices are pairs $(M,S)$ where $M$ is a perfect matching and $S$ is any subset of the unmatched edges. (This is essentially a dynamic programming algorithm but we find it more convenient to describe it as a graph search.) In our search, we will use $M$ to represent the matching at the current endpoint of a partial sequence of matchings of $G$, and $S$ to represent the edges that have already been covered by matchings earlier in the partial sequence. We make an edge in $\Gamma$ from vertex $(M,S)$ to vertex $(M',S')$ when matchings $M$ and $M'$ differ by a single rotation, when $M'\cap S=\emptyset$, and when $S'=S\cup(M\setminus M')$. With this definition, paths in $\Gamma$ automatically correspond to sequences of matchings that never re-use any edges: an edge, once used, gets added to $S$ and its presence there and the requirement that $S$ be disjoint from $M'$ prevent later matchings in the sequence from re-using it. Therefore, the sequences we seek are exactly the paths from any vertex of the form $(M,\emptyset)$ to any other vertex of the form $(M',E(G)\setminus M')$. We can find such a path by breadth first or depth first search of $\Gamma$, starting from the vertices $(M,\emptyset)$ and terminating whenever a vertex of the form $(M',E(G)\setminus M')$ is reached.

By \autoref{lem:exponentials}, $\Gamma$ has $O(c_1^m)$ vertices and $O(c_2^m)$ edges.
It remains to describe how to construct $\Gamma$ and search it in a space-efficient way, bearing in mind that we do not want to use the amount of memory it would require to store all the edges of $\Gamma$. We have enough memory that (by the assumption that machine words are large enough to store a single memory address) we can store a single matching $M$ or a single set of edges $S$ in a single word, and perform set operations like unions and intersections of these sets in constant time per operation. We will index the vertices of $\Gamma$ by consecutive integers in the range from $0$ to $|V(\Gamma)|-1$. Based on these indexes, we store the following information:
\begin{itemize}
\item An array $A_1$, addressed by vertex index, of the pairs $(M_i,S_i)$ of structures associated with vertex $i$.
\item An array $A_2$, addressed by vertex index, of the predecessor of vertex $i$ in the first path to reach that vertex found by the search, or a special flag value to indicate that it has not yet been reached.
\item An array $A_3$ of indexes of vertices on the current frontier of the search, arranged either as a stack for depth first search or a queue for breadth first search.
\item An array $A_4$ of Boolean values indexed by sets of edges, true when that set of edges is a simple cycle and false otherwise.
\item An array $A_5$ of lists indexed by sets $S$ of edges, listing the matchings $M$ disjoint from $S$ in sorted order, together with the vertex index for $(M,S)$ for each matching.
\end{itemize}
Array $A_4$ takes space $2^m$ to store, and the remaining arrays (including the total length of the lists in $A_5$ take space $O(c_1^m)$. Initializing $A_1$ can be accomplished in linear time using the structure enumeration algorithm of the proof of \autoref{lem:exponentials}. Array $A_2$ is easily initialized to contain only flag values, and $A_3$ is initialized to list all vertices of the form $(M,\emptyset)$. Initializing $A_4$ can be done by a naive algorithm that tests each set for being a cycle, in time $O(m2^m)$, smaller than our stated time and space bounds. $A_5$ can be initialized by using $A_1$ to count how many matchings are associated with each disjoint set in order to allocate storage for the lists within $A_5$,  using a second pass over $A_1$ to place matchings into lists, and then sorting each list.

Each step of the search of $\Gamma$ is obtained by removing a vertex $(M,S)$ of $\Gamma$ from the search frontier $A_3$, terminating the algorithm if $M\cup S=E(G)$, otherwise looping through the neighbors of $(M,S)$ in $\Gamma$, and for each previously-unreached neighbor setting its predecessor in $A_2$ and adding it to the frontier $A_3$. It remains to describe how to find the outgoing neighbors of a vertex $(M,S)$ of $\Gamma$, since these are not explicitly represented by the data structures described above (and storing them explicitly would use too much space). We do this simply by using $A_5$ to find all of the matchings $M'$ disjoint from $S$, performing set operations on $M$ and $M'$ to find their symmetric difference, and using $A_4$ to check that this symmetric difference is a simple cycle. When it is, we have found a neighbor $(M',S')$ of $(M,S)$ in $\Gamma$ , with $S'=S\cup(M\setminus M')$. We can find the vertex index of this neighbor by a binary search in the list for $S'$ of $A_5$. The bottleneck of the algorithm is performing these binary searches; each takes time $O(m)$ and the number of searches performed is~$O(c_2^m)$.

The alternative space and time bounds of $O(n! 2^{m-n})$ and $O(n!^2 2^{m-n})$, which are better for dense graphs,
follow from the same algorithm (with a preprocessing stage that eliminates vertices of degree less than two) using a more naive analysis in which we count pairs of a matching and a disjoint set, or two matchings and a disjoint set, by bounding the number of matchings by $(n/2)!$ and the number of disjoint sets by two to the power of the number of unmatched edges.
\end{proof}

By combining this algorithm with a kernelization algorithm that removes isolated vertices or edges,
immediately halts if there are other degree-one vertices, and compresses all paths of degree-two vertices to paths of the same parity with one or two degree-two vertices, we would also obtain an algorithm that is fixed-parameter tractable in the number of vertices of degree greater than two. However, this is subsumed by the results of the next section, in which we show that the problem is fixed-parameter tractable in the carving width of the input graph.

\subsection{Carving width}

In this section we parameterize the recognition of graphs of stably matchable pairs by \emph{carving width}.
The precise definition of carving width involves hierarchical clusterings of the vertices of a graph so that each cluster boundary is crossed by a small number of edges, but we do not need the details for this section; instead, it suffices to know that for graphs of treewidth $w$ and maximum degree $d$, the carving width is bounded below by $\Omega(\max(w,d))$ and bounded above $O(wd)$~\cite{Epp-JGAA-18}. Therefore, we can equivalently think of the graphs of small carving width as being graphs that simultaneously have small treewidth and small maximum degree.

As is standard for fixed-parameter tractable algorithms involving treewidth, our treewidth-parameterized algorithm for recognizing graphs of stably matchable pairs uses dynamic programming on a \emph{tree decomposition}, a tree whose nodes are associated with sets of vertices of the given graph (called \emph{bags}), with the constraints that every edge of the graph has both endpoints in at least one bag and that the bags containing any given vertex of the graph are associated with the nodes of a connected subtree. The \emph{width} of the decomposition is one less than the number of vertices in its largest bag, and the treewidth of the graph is the minimum width of a tree decomposition. It will be convenient to use a \emph{nice tree decomposition}, in which the tree of bags is rooted and each bag has one of three types~\cite{DorTel-DAM-09}:
\begin{itemize}
\item An \emph{introduce node} is either a leaf node whose bag contains a single vertex, or a node with exactly one child whose bag is formed by adding one vertex to its child's bag.
\item A \emph{forget node} is a node with exactly one child whose bag is formed by removing one vertex from its child's bag.
\item A \emph{join node} has exactly two children, with the bags of it and its children all being equal.
\end{itemize}
Each vertex can be introduced in many introduce nodes, but forgotten only from one, which must be the parent of the common ancestor of the introduce nodes. It will be convenient to introduce one more type of node, an \emph{edge node}, associated bijectively with a specific edge in the given graph whose endpoints belong to the bag of the edge node. In this way, each node of the tree decomposition can be associated with a unique subgraph of the given graph, the subgraph whose edges and vertices are associated with the given node and its descendants. By introducing additional forget nodes we may assume without loss of generality that the root of the tree has an empty bag; its associate is the whole graph. We describe a tree decomposition that meets these minor modifications to the usual definition of a nice tree decomposition as a \emph{good tree decomposition}.

Minimum-width tree decompositions can be found in time linear in the number of vertices of an input graph but exponential in the cube of the width~\cite{Bod-SICOMP-96}; it is also possible to find an approximate tree decomposition whose width is within a constant factor of minimum, in time linear in the number of vertices of the input graph and single-exponential in the width~\cite{BodDraDre-SICOMP-16}. An arbitrary tree decomposition can be transformed into a nice tree decomposition of the same width, blowing up the number of nodes by a factor polynomial in the width, in time linear in the number of vertices of the graph and polynomial in the width~\cite{DorTel-DAM-09}. Adding edge nodes between the highest node whose bag contains both endpoints of each edge and the forget nodes above them introduces a number of nodes which is again linear in vertices and polynomial in width, within the same time bounds.

Our dynamic programming algorithm will consider two different kinds of partial information about a rotation system for the given graph, restricted to the subgraph associated with a node of a nice tree decomposition, which we call a \emph{rotation subsystem} and a \emph{rotation state}. A rotation subsystem completely describes the part of a rotation system within the subgraph associated with a node, while a rotation state describes only the parts of the rotation system involving the vertices of the current bag. A rotation subsystem and a rotation state are \emph{compatible} when the rotation state is obtained from the rotation subsystem by forgetting the information about vertices that are not in the bag. For a given node $x$ of the tree decomposition, let $G_x$ denote the subgraph associated with that node, $B_x$ denote the set of vertices of the bag associated with that node, $F_x=V(G_x)\setminus B_x$ denote the set of forgotten vertices, and $G_x[B_x]$ denote the part of $G_x$ having both endpoints in $B_x$. More specifically, a rotation subsystem for a node $x$ consists of the following information:
\begin{itemize}
\item Top and bottom matchings $\top\cap G_x$ and $\bot\cap G_x$ in which all vertices of $F_x$ are matched, but vertices of $B_x$ may or may not be matched.
\item A collection of rotations, an abstract partially ordered set whose elements will represent the subset of rotations of the rotation system that include at least one edge of $G_x$.
\item A collection of \emph{pieces} of rotations, the intersections of rotations with $G_x$. Each piece is either a full cycle, or a connected path whose endpoints both belong to $B_x$. Each two pieces of the same rotation must be disjoint. A rotation can either have a full cycle as its piece or one or more paths as pieces, but not both.
\end{itemize}
Each edge in the subgraph associated with a node must be part of two elements (either pieces or $\top$ or $\bot$), and be upper in one and lower in the other (as determined by the partial order of the rotations). Each piece must alternate between upper and lower edges. We do not require that this structure can be extended to a full rotation system, in general; however, we have:

\begin{observation}
A rotation subsystem at the root node of a good tree decomposition coincides with a rotation system for the given graph.
\end{observation}

\begin{proof}
Rotation subsystems differ from rotation systems in allowing pieces of rotations that are paths rather than cycles, and in allowing the bag vertices to be unmatched in $\top$ and $\bot$. But because the root node of a good tree decomposition is assumed to be empty bag, these differences are not possible at that node.
\end{proof}

We define a \emph{rotation state} at a node of the tree decomposition to consist of the following information, corresponding to the information in a rotation subsystem but restricted to bag vertices. Then a rotation state for $x$ consists of:
\begin{itemize}
\item Top and bottom matchings $\top$ and $\bot$ in $G_x[B_x]$, together with two bits of information for each remaining bag node specifying whether it has a match in $\top\cap G_x$ or $\bot\cap G_x$.
\item A collection of rotations, an abstract partially ordered set whose elements will represent the subset of rotations of the rotation system that either include at least one edge of $G_x$ or at least one piece that is a path with an endpoint in $B_x$. Because of the assumption that the input graph has low degree, the number of rotations in this collection will also be low.
\item A collection of pieces of rotations. For each piece we store the rotation that it is part of, the edges of the piece that belong to $G_x[B_x]$, and the endpoints of the piece in $B_x$, but not the edges in the rest of $G_x$. At each endpoint in $B_x$, we store whether the edge of the piece that is incident to that endpoint is an upper or lower edge in its rotation.
\end{itemize}
We require that each edge of $G_x[B_x]$ belong to two elements (either pieces or $\top$ or $\bot$), and be upper in one and lower in the other, and that when two edges of $G_x[B_x]$ are consecutive within a piece that they alternate between upper and lower. We do not require that this structure can be extended to a rotation subsystem for the same node of the tree decomposition, but we call it \emph{valid} when it can be so extended.

\begin{lemma}
\label{lem:nstates}
In a tree decomposition of width $w$ for a graph of degree $d$, the number of possible rotation states is at most exponential in $O((wd)^2)$.
\end{lemma}

\begin{proof}
Because each piece of rotation represented within a rotation state uses at least one of the incident edges at one of the vertices of $B_x$, there can be at most $O(wd)$ pieces, and $O(wd)$ rotations. The partial order on rotations can be represented using $O((wd)^2)$ bits of information, the top and bottom matchings in $G_x[B_x]$ require $O(w^2)$ bits of information, and the information about which rotation pieces belong to which rotations, which bag vertices are endpoints of which pieces, and whether the incident edges are upper or lower, all take $O(wd\log(wd))$ bits of information.
\end{proof}

Our overall algorithm will construct a tree decomposition and then work bottom up through its nodes, calculating at each node the set of valid rotation states. To do so, it uses the following lemmas, which describe the sets of valid rotation states for each possible type of node in our decomposition, and allow each set to be computed in an amount of time that is exponential in a polynomial of the carving width, but independent of the overall graph size.

\begin{lemma}
At an introduce node $x$ of the tree decomposition, with child $y$, a state $S$ is valid if and only if the vertex $v$ that is introduced at $x$ does not take part in any pieces of rotations, $v$ is marked as not matched in $\top\cap G_x$ and $\bot\cap G_x$, and the state $S'$ formed by removing $v$ from $S$ vertex is valid for $y$.
\end{lemma}

\begin{proof}
If $S$ is valid, it has a rotation subsystem for $x$, which cannot use vertex $v$ in any rotation or in the top or bottom matching because $v$ has no edges in $G_x$. Therefore, removing $x$ produces a rotation subsystem for $y$, showing that $S'$ is valid. Conversely, if $S'$ is valid, it has a rotation subsystem for $y$ Adding $v$ to this subsystem without changing its rotations or pieces produces a rotation subsystem for $x$, in which $v$ is unmatched, showing that $S$ is valid.
\end{proof}

\begin{lemma}
At a forget node $x$ of the tree decomposition, with child $y$, a state $S$ is valid if and only if the vertex $v$ that is forgotten at $x$ does not take part in any pieces of rotations, and there exists a valid state $S'$ formed from $S$ by adding $v$ to $S$, marking it as matched in both $\top\cap G_y$ and $\bot\cap G_y$, and possibly adding an edge incident to $v$ in $G_y[B_y]$ to $\top$ or $\bot$.
\end{lemma}

\begin{proof}
If $S$ is valid, it has a rotation subsystem for $x$, which cannot have vertex $v$ as the endpoint of a piece of a rotation because $v$ does not belong to $B_x$. Because $x$ belongs to $F_x$, it must be matched in both $\top$ and $\bot$ in this rotation subsystem. The same rotation subsystem exists for $y$, and corresponds to a state $S'$ that meets the conditions of the lemma. Conversely, if a state $S'$ as described in the lemma exists, it has a rotation subsystem, which matches $v$ in both $\top$ and $\bot$. This matching is the only additional requirement for a rotation subsystem at $x$, so there is a rotation subsystem at $x$ showing that $S$ is valid.
\end{proof}

\begin{lemma}
At an edge node $x$ which introduces edge $uv$ as the only edge of graph $G_x$, a state $S$ is valid if and only if one of the following three conditions is met:
\begin{itemize}
\item Edge $xy$ is included in $\top$, vertices $u$ and $v$ are marked as matched in $\top$ and unmatched in $\bot$, and there are no rotations or pieces of rotations.
\item Edges $xy$ is included in $\bot$, vertices $u$ and $v$ are marked as matched in $\top$ and unmatched in $\bot$, and there are no rotations or pieces of rotations.
\item Edge $xy$ is not included in $\top$ or $\bot$, vertices $u$ and $v$ are marked as unmatched in both $\top$ and $\bot$, and there is a single rotation having a single piece, consisting of this edge, marked as either upper or lower.
\end{itemize}
\end{lemma}

\begin{proof}
In each case $G_x=G_x[B_x]$, so rotation subsystems coincide with rotation states, and these are the only possible rotation subsystems for a single-edge graph.
\end{proof}

The next lemma lets us reduce the analysis of edge nodes more generally to be a special case of join nodes.

\begin{lemma}
At an edge node $x$ with child $y$ which introduces edge $uv$ to a non-empty subgraph $G_y$, a state $S$ is valid if and only if it would be valid for a join node $x'$ with the same bag as $x$ whose two children are $y$ and an edge introducing $uv$ as the only edge in its subgraph.
\end{lemma}

\begin{proof}
Validity of a rotation state at a node $x$ depends only on $G_x$ and $B_x$, and for the nodes described in the lemma we have $G_x=G_{x'}$ and $B_x=B_{x'}$.
\end{proof}

\begin{lemma}
At a join node $x$ with children $y$ and $z$, a rotation state $S$ is valid if and only if there exist valid states $S'$ for $y$ and $S''$ for $z$ such that:
\begin{itemize}
\item The matchings $\bot$ and $\top$ in $S$ are disjoint unions of the corresponding matchings in $S'$ and $S''$, and the sets of matched vertices in $\bot$ and $\top$ in $S$ are disjoint unions of the corresponding sets of matched vertices in $S'$ and $S''$.
\item The rotations in $S$ are a union (not necessarily disjoint) of the rotations in $S'$ and $S''$, with the partial order on these rotations in $S'$ and $S''$ equalling the restriction to those subsets of the partial order on rotations in $S$.
\item Each piece of a rotation in $S$ is formed as a concatenation of one or more pieces of rotations in $S'$ and $S''$, all pieces of rotations  in $S'$ and $S''$ are used to form pieces of rotations in this way, and at each point of concatenation the two concatenated pieces alternate between an upper edge and a lower edge.
\end{itemize}
\end{lemma}

\begin{proof}
If $S$ is valid, it is associated with a rotation subsystem, and restricting that subsystem to $G_y$ and $G_z$ produces two subsystems with disjoint matchings $\bot$ and $\top$, each having a subset of the rotations with a restriction of the partial order on rotations, and each having pieces of rotations that are subsets of the pieces of rotations in the subsystem for $x$. It follows that the rotation states $S'$ and $S''$ associated with those restricted rotation subsystems are valid and meet the conditions of the lemma. Conversely, if valid rotation states $S'$ and $S''$ meeting the conditions of the lemma exist, then they each are associated with rotation subsystems for $y$ and $z$ that can be combined together in the obvious way to produce a rotation subsystem for $x$, showing that $S$ is valid.
\end{proof}

\begin{theorem}
\label{thm:cw-alg}
For a graph $G$ with $n$ vertices and carving width $w$, we can test whether $G$ is a graph of stably matchable pairs in time linear in $n$ and exponential in $O(w^4)$.
\end{theorem}

\begin{proof}
We can find a good tree-decomposition of width $O(w)$ in the stated time bound. A graph with width $O(w)$ has $O(wn)$ edges, so there are $O(wn)$ edge nodes in the tree-decomposition. We may eliminate any join nodes in which one of the two sides of the join has no edges, without change to the validity of the decomposition, leaving a decomposition that also has $O(wn)$ join nodes and (as for any nice tree decomposition) $n$ forget nodes. Removing the forget nodes and join nodes partitions the tree-decomposition into subtrees within which there can be at most $O(w)$ introduce nodes, so there can be at most $O(w^2n)$ introduce nodes total.

Because $G$ has carving width $w$, its treewidth and maximum degree are both $O(w)$, so by \autoref{lem:nstates} the number of states at each node of this decomposition is exponential in $O(w^4)$. By the lemmas above, we can compute the set of valid states at each node by an algorithm that examines each valid state at the child node (when there is one child) or each pair of valid states at the two children (when there are two children), in an amount of time per state or pair of states that is polynomial in $w$. Therefore the total time per node is bounded by this polynomial multiplied by the square of the number of states, which remains exponential in $O(w^4)$.
\end{proof}

\section{Conclusions and open problems}

We have investigated the graphs of stably matchable pairs of stable matching instances, with results including the characterization of special classes of these graphs, characterization of the lattices of stable matchings corresponding to these classes of graphs, the NP-completeness of recognizing graphs of stably matchable pairs in general, and exact algorithms for recognizing these graphs that run in singly exponential or fixed-parameter tractable time.

A natural open problem that remains concerns whether our algorithm for carving width can be extended to treewidth for graphs of unbounded degree. As a first step towards  this, what is the complexity of recognizing series-parallel graphs of stably matchable pairs? Another problem concerns the analysis of numbers of combinations of matchings and subsets of edges in \autoref{lem:exponentials}, which we used to analyze our exponential-time recognition algorithm. This analysis seems unlike to be tight; can it be improved?

\bibliographystyle{plainurl}
\bibliography{smg}
\end{document}